\documentclass[11pt]{article}%
\usepackage{amsfonts}
\usepackage{fancyhdr}
\usepackage{comment}
\usepackage{amsmath}
\usepackage{amsthm}
\usepackage{amsmath}
\usepackage{amssymb}
\usepackage{amssymb}
\usepackage{natbib}
\usepackage[left=1in, top=1in, right=1in, bottom=1in]{geometry}
\usepackage{authblk}

\usepackage{booktabs} 

\usepackage{mdframed}
\usepackage{anyfontsize}
\usepackage{natbib}
\usepackage{makecell}
\usepackage{graphicx}%
\usepackage{bbm}
\usepackage{hyperref}
\usepackage{cleveref}
\usepackage{url}
\usepackage{appendix}
\usepackage{xcolor}

\setcitestyle{authoryear}

\newtheorem{theorem}{Theorem}

\newtheorem{algorithm}[theorem]{Algorithm}

\newtheorem{definition}[theorem]{Definition}

\newtheorem{lemma}[theorem]{Lemma}

\newtheorem{remark}[theorem]{Remark}

\newcommand{\N}{\mathbb{N}}
\newcommand{\R}{\mathbb{R}}

\newcommand{\E}{\mathbb{E}}
\usepackage{graphicx}
\graphicspath{ {./images/} }

\newcommand{\Feasibility}{\mathcal{F}}
\newcommand{\I}{\mathcal{I}}
\newcommand{\vv}{\varphi}
\newcommand{\strategy}{\pi}
\newcommand{\critical}{\beta}

\newcommand{\BidSpace}{\mathcal{B}}

\newcommand{\activeBid}{\mathcal{A}}
\newcommand{\competeBid}{\mathcal{C}}
\newcommand{\promiseBid}{\mathcal{P}}

\newcommand{\ADRA}{\text{ADRA}}
\newcommand{\vb}{\boldsymbol}
\renewcommand{\vec}{\boldsymbol}

\newcommand{\comm}{\mathsf{commit}}
\newcommand{\OPT}{\mathsf{OPT}}

\newcommand{\mechanism}{\mathcal{M}}

\newcommand{\FullRank}{\xi}
\newcommand{\EmptyRank}{\overline{\FullRank}}
\newcommand{\incr}{\Gamma}

\title{Truthful, Credible, and Optimal Auctions for Matroids via Blockchains and Commitments}

\author[1]{Aadityan Ganesh}
\author[2]{Qianfan Zhang}
\affil[1]{Princeton University, \texttt{aadityanganesh@princeton.edu}}
\affil[2]{Princeton University, \texttt{qianfan@princeton.edu}}

\date{}

\begin{document}


\maketitle


\begin{abstract}
    We consider a revenue-optimizing auctioneer in single-dimensional environments with matroid feasibility constraints.
    \citet{AL20} argue that any revenue-optimal, truthful, and credible mechanism requires unbounded communication.
    Recent works \citep{FW20, EFW22, CFK24} circumvent their impossibility for the single-item setting through the use of cryptographic commitments and blockchains.
    We extend their results to matroid feasibility constraints.

    At a high level, the two-round Deferred-Revelation Auction (DRA) discussed by \citet{FW20} and \citet{CFK24} requires each bidder to submit a deposit, which is slashed upon presenting verifiable evidence indicating a deviation from the behaviour prescribed by the mechanism.
    We prove that the DRA satisfies truthfulness, credibility and revenue-optimality for all matroid environments when bidders' values are drawn from $\alpha$-strongly regular distributions for $\alpha > 0$.
    Further, we argue that the DRA is not credible for any feasibility constraint beyond matroids and for any smaller deposits than suggested by previous literature even in single-item environments. 

    Finally, we modify the Ascending Deferred-Revelation Auction (ADRA) for single-item settings proposed by \citet{EFW22} for arbitrary bidder value distributions.
    We implement a deferred-revelation variant of the deferred-acceptance auction for matroids due to \citet{BdVSV11}, which requires the same bounded communication as the ADRA.
\end{abstract}


\section{Introduction}

    We study revenue optimization in single-dimensional environments with matroid feasibility constraints and an untrusted auctioneer.
    Traditional mechanism design literature stemming from the seminal work of \citet{Myerson81} has been studying revenue-optimal mechanisms where bidders are incentivized to not be strategic and truthfully report their values to the mechanism.
    However, most of the literature assumes that the auctioneer can commit to running a mechanism and does not deviate from the committed mechanism (by, for example, fabricating bids from a hypothetical bidder) even if it results in a larger revenue to the auctioneer.
    A recent line of work has been considering a strategic auctioneer who lacks the ability to guarantee sticking to its commitments \citep{DPS14, ILPT17, LMSZ19}.
    
    For one-shot auctions, \citet{AL20} coin the term \emph{credibility} to denote auctions where the auctioneer need not be trusted to honestly implement the mechanism.
    However, they argue that credibility introduces a trilemma --- without making additional assumptions, no revenue-optimal credible mechanism that is also truthful for bidders can terminate in a bounded number of rounds.
    A research agenda has been developing around side-stepping their trilemma for single-item environments by making additional computational assumptions on the auctioneer --- such as the existence of secure cryptography that cannot be broken by the auctioneer in polynomial time \citep{FW20, EFW22} and making very mild trust assumptions (much milder than having to fully trust the auctioneer) such as the existence of a trusted public ledger implemented via a blockchain \citep{CFK24}.
    We extend the results from these works to matroid environments.
    
    We look at the Deferred-Revelation Auction (DRA; \citealp{FW20}) and the Ascending Deferred-Revelation Auction (ADRA; \citealp{EFW22}) for matroid feasibility constraints.
    At a high level, the DRA and the ADRA implement the revenue-optimal DSIC sealed-bid auction \citep{Myerson81} and the ascending-price auction \citep{BdVSV11} for matroids respectively.
    However, in order to participate in either of the two auctions, bidders will have to cryptographically commit to their bids at the start of the auction and reveal them only after the auctioneer confirms receipt of the commitments and notifies the bidders of all the commitments that it has received.
    Additionally, each committed bid will also have to lock-in a collateral, which is slashed if a bidder behaves erratically.
    This prevents the auctioneer from, for example, fabricating bids based on the bids received from the bidders or not revealing some its own fake bids. 
    
    Similar to \citet{CFK24}, we implement the DRA and ADRA over a blockchain.
    If the auctioneer communicates with bidders through private channels, it can be strategic by fabricating bids for a select subset of bidders, or getting a few bidders to reveal their bids first (by notifying them of the commitments received thus far) before even soliciting bids from the rest.
    Bidders have no way to verify whether the auctioneer's private messages that they have received are consistent with those sent to others.  
    A public ledger, such as a blockchain, resolves this issue by ensuring that all communications between parties are recorded transparently for everyone to see.

    However, moving from single-item environments to matroids vastly expands the auctioneer's strategy space.
    For example, when selling $k$ identical items, the auctioneer can announce an auction selling $k+1$ items as long as it ensures one of its fabricated bids ends up winning an item.
    More generally, as long as the final allocation to the bidders is feasible, the auctioneer can strategically announce a false feasibility constraint.
    
    Alongside fancier feasibility constraints, we depart from prior work by considering bidders' values drawn from independent, but not necessarily identical distributions.    
    The DRA and ADRA are primarily designed with pseudonymous environments in mind, where bidders are not directly identifiable. However, the auctioneer may still be able to infer certain attributes based on the pseudonymous identities, such as the level of activity associated with an identity in an NFT community.
    The prior for a bid could be a function of its publicly observable attributes.
    However, this enables the auctioneer to strategically pick distributions for its fake bids so as to maximize its revenue.

    For matroid feasibility constraints, we prove that the same collateral prescribed by \citet{FW20} and \citet{CFK24} for the DRA when bidders' values are drawn from MHR and $\alpha$-strongly regular distributions respectively and the collateral prescribed by \citet{EFW22} for arbitrary value distributions for the ADRA are sufficient for credibility.
    Our results for the DRA on a public ledger are tight in the following aspects --- for MHR distributions, the DRA is not credible for any smaller collateral than recommended by \citet{FW20} and is not credible for any non-matroid feasibility constraint.
    Further, we argue that relaxing the communication model from public to private communication will require a larger collateral for matroid constraints beyond single-item settings.

    Our work is applicable to multi-unit auctions where the auctioneer can strategically disclose its total supply of units.
    For instance, partition matroids naturally arise when items must be allocated across distinct bidder populations--- such as regional markets, customer segments, or product lines--- each with its own allocation limit.
    The DRA and ADRA allow implementing credible mechanisms in such environments.

\subsection{Related Works}

Our work fits in the credibility framework introduced by \citet{AL20}, with additional cryptography assumptions.
\citet{FW20, EFW22} and \citet{CFK24} are the works closest to ours in prior literature --- they focus on iid single-item auctions while we consider non-iid matroid constraints.
More broadly, our work would fit into the literature on mechanisms with imperfect commitments.
\citet{DFLSV20} consider designing simple, credible and approximately optimal but not necessarily truthful auctions in multi-item environments.
\citet{LMSZ19} and \citet{skreta15}, on the other hand, consider repeated auctions, where the auctioneer can commit to its behaviour in the current iteration of the auction, but not in the future iterations.
Our work, on the other hand, focuses on designing truthful single-shot auctions.

Recent works in the transaction fee mechanism (TFM) design literature have been considering designing cryptographically-assisted truthful, credible, collusion-resistant, and off-chain influence proof mechanisms.
While \citet{Roughgarden24} initiates the study of TFMs to decide inclusion of transactions in a block, \citet{SCW23} and \citet{GTW24} consider TFMs armed with cryptography.
\citet{SCW23} assumes a secure multi-party computation where the auctioneer cannot conceal any bids once submitted.
\citet{GTW24} in contrast, use the DRA as a black-box --- they (implicitly) assume a sufficiently high collateral so that the auctioneer will never be incentivized to conceal any of its bid after committing to them.
However, both of these works assume the implementation of credible-cryptographic auctions as a black-box and do not discuss the exact implementation details.

A broad line of works considers ascending-price auctions for various feasibility constraints.
\citet{MS14} formally model the class of all ascending-price mechanisms by deferred-acceptance auctions.
Near-optimal (for surplus and revenue) deferred-acceptance auctions for various environments are known \citep{BdVSV11, BGGS22, CGS22, FGGS22}.
Our work does not focus on designing new deferred-acceptance auctions, but rather on reducing the communication complexity of the auction discussed in \citet{BdVSV11} through cryptographic commitments.   

\subsection{Paper Organization}

We begin by formally describing the DRA for matroid feasibility constraints in \Cref{sec:DRA}.
We spread the discussion of our extension to the DRA over \Cref{sec:WarmUp} and \Cref{sec:UpperBound}.
In \Cref{sec:WarmUp}, we revisit the single-item setting from \citet{FW20} and discuss a modified version of their proof that acts as a natural stepping-stone for matroid environments.
We then extend the ideas from the modified proof to give upper bounds on the sufficient collateral for the DRA to be credible when the value distributions are MHR in \Cref{sec:MatroidMHR} and when the distributions are $\alpha$-strongly regular in \Cref{sec:MatroidReg}.
We then discuss our lower bounds in \Cref{sec:LowerBounds} --- we argue that the DRA is not credible for (a) feasibility constraints beyond matroids, (b) any collateral smaller than the one prescribed by \citet{FW20} when the value distributions are MHR and (c) for a private communications, instead of using a public ledger.
Finally, we describe the ADRA for matroids and argue its credibility in \Cref{sec:ADRA}.

\Cref{tab:comparison} compares prior work with our results along three axes: communication channel, assumptions on users' value distributions, and feasibility constraints.

\begin{table}[h]
\centering
\begin{tabular}{l|lll}
& Channel & Distributional Assumption & Feasibility\\ \hline
\citet{FW20} & Private & MHR & Single-item \\
\citet{EFW22} & Private & General & Single-item \\
\citet{CFK24} & Public & MHR / $\alpha$-strongly regular & Single-item \\
\Cref{sec:UpperBound} & Public & MHR / $\alpha$-strongly regular & Matroid \\
\Cref{sec:ADRA} & Public & General & Matroid
\end{tabular}
\caption{Summary of prior work and our contributions.}
\label{tab:comparison}
\end{table}

\section{Preliminaries}

We study single-dimensional environments with quasi-linear bidders for matroid feasibility constraints.
We begin the preliminaries by reviewing matroids before proceeding to the primitives from the mechanism design literature.

\subsection{Matroid Environments}
A feasibility constraint $\Feasibility \subseteq 2^{\N}$ dictates the set of bidders that can be allocated simultaneously by the mechanism.
A matroid environment is an environment given by a matroid feasibility constraint.

\begin{definition}[Matroids] \label{def:Matroids}
    A matroid $M(E, \I)$ is given by a set of elements $E$ and a family of independent sets $\I \subseteq 2^E$ such that:
    \begin{enumerate}
        \item (Downward-closure) If $W \in \I$, then for all subsets $\hat{W}$ of $W$, $\hat{W} \in \I$.
        \item (Augmentation property) If $W$ and $\hat{W}$ are two independent sets such that $|W| < |\hat{W}|$, then, there exists an element $w \in \hat{W} \setminus W$ such that $W \cup \{w\} \in \I$.
    \end{enumerate}
    The rank of $M$ is the cardinality $|I|$ of the largest independent set $I \in \I$. 
\end{definition}

\Cref{thm:Greedy} is a generalization of the augmentation property.

\begin{lemma}[See \citealp{Oxley2006}, for example] \label{thm:Greedy}
    Let $M(E, \I)$ be a matroid. Let $W$ and $\hat{W}$ be two independent sets in $\I$ with $|W| < |\hat{W}|$.
    Then, there exists $\hat{D} \subseteq \hat{W} \setminus W$ such that $W \cup \hat{D} \in \I$ and $|W \cup \hat{D}| = |\hat{W}|$.
\end{lemma}

\subsection{Auctions and Incentive Compatibility}

\subsubsection{Mechanisms}

For a feasibility constraint $\Feasibility$ and a set of $n$ bidders, an auction (or a mechanism) is described by a pair $(X, p)$ that takes as input a vector of bids $(b_1, \dots, b_n)$ placed by the $n$ bidders, each belonging to some bid space $\BidSpace$, and maps it to a distribution over feasible allocations and a vector of payments charged to each bidder.
For an allocation rule $X$, let $x = (x_1, \dots, x_n)$ denote the allocation probabilities for bidders $1, \dots, n$ respectively, i.e, for bids $\vec{b}$ placed by the bidders, $x_i(\vec{b}) = Pr_{Y \sim X(\vec{b})}(i \in Y)$.
We will use the term allocation rule for both $X$ and $x$.

Each bidder $i$ has a private value $v_i$ that is unknown to the auctioneer and the other bidders drawn from a publicly known distribution $D_i$.
A strategy $s_i:\R_{\geq 0} \xrightarrow{} \BidSpace$ for bidder $i$ maps the bidder's value $v_i$ to its bid $b_i$ in the auction.
We assume that bidders are quasi-linear --- bidders place bids so as to maximize their \emph{utility} $v_i \, x_i - p_i$ from being allocated with probability $x_i$ and charged $p_i$ by the auction.

A mechanism is \emph{direct} if it asks bidder to report their values directly, i.e, $\BidSpace = \R_{\geq 0}$.
Otherwise, the mechanism is said to be \emph{indirect}.
For example, the first-price auction that collects bids, allocates the highest bidder and charges the highest bidder its bid is direct.
The ascending-price auction, where the price for the item gradually increases starting at zero until only one bidder is willing to pay the price, is indirect --- every time the price increases, bidders bid ``yes'' if they are willing to continue in the auction and ``no'' otherwise.

\subsubsection{Incentive Compatibility}

At a high level, a mechanism is incentive compatible if there exists some ``prescribed strategy'' for each bidder that they can follow to optimize their utility, irrespective of the values of the other bidders.

A direct mechanism is \emph{dominant-strategy incentive compatible} (DSIC) if each bidder optimizes their expected utility irrespective of the bids placed by the other bidders.
Formally, an allocation rule $x$ and a payment rule $p$ is DSIC if
$$v_i \, x_i(v_i, b_{-i}) - p_i(v_i, b_{-i}) \geq v_i \, x_i(b_i, b_{-i}) - p_i(b_i, b_{-i})$$
for all bidders $i$, bids $v_i, b_i \geq 0$ of bidder $i$ and bids $b_{-i}$ of bidders other than $i$.

Similarly, an indirect mechanism is \emph{ex-post incentive compatible} (EPIC) if for some prescribed strategies $s_1, \dots, s_n$ for the $n$ bidders, bidder $i$ maximizes its utility by bidding accordingly to $s_i(v_i)$ as long as the other bidders bid according to the strategies $s_{-i}$, but can choose to misrepresent their values to the mechanism.
Formally, a mechanism with an allocation rule $x$ and payment rule $p$ is EPIC if
$$v_i \, x_i(s_i(v_i), s_{-i}(b_{-i})) - p_i(s_i(v_i), s_{-i}(b_{-i})) \geq v_i \, x_i(s_i(b_i), s_{-i}(b_{-i})) - p_i(s_i(b_i), s_{-i}(b_{-i}))$$
for all bidders $i$, bids $s_i(v_i), s_i(b_i) \in \BidSpace$ of bidder $i$ and bids $s_{-i}(b_{-i})$ of bidders other than $i$.

For an indirect mechanism $(x, p)$ and EPIC strategies $s_1, \dots, s_n$, the revelation principle states that the \emph{direct-revelation mechanism} given by $\Tilde{x}(\vec{v}) = x(s_1(v_1), \dots, s_n(v_n))$ and $\Tilde{p}(\vec{v}) = p(s_1(v_1), \dots, s_n(v_n))$ is DSIC.

\begin{theorem}
    The direct-revelation mechanism $(\Tilde{x}, \Tilde{p})$ of an indirect EPIC mechanism $(x, p)$ is DSIC.
\end{theorem}

\subsubsection{Myerson's Payment Identity}
For an allocation rule $x$, \citet{Myerson81} argued that there exists a payment rule $p$ such that $(x, p)$ is DSIC if and only if $x$ is monotone, i.e, for all bidders $i$ and bids $b_{-i}$ of the other bidders, $x_i(b_i, b_{-i})$ is monotone non-decreasing in $b_i$.
Further, through the payment identity, he characterizes the payment rule $p$ for which $(x, p)$ is DSIC.
Throughout this paper, we focus exclusively on deterministic mechanisms and, accordingly, present a version of the payment identity tailored to deterministic auctions.

We approach the payment identity via \emph{critical bids}.
For a deterministic allocation rule $x$ and bids $b_{-i}$, the critical bid $\critical_i(x, b_{-i})$ for bidder $i$ is the minimum bid to be placed by bidder $i$ so that $x_i(b_i, b_{-i}) = 1$ (by monotonicity, note that $x_i(v_i, b_{-i}) = 1$ for all $v_i \geq \critical_i(x, b_{-i})$).
The payment identity suggests charging the critical bid $\critical_i(x, b_{-i})$ from each allocated bidder $i$ for a deterministic mechanism $x$.
We state the payment identity assuming \emph{individual rationality} and \emph{no positive transfers}, i.e, a bidder guarantees itself no payments by bidding zero.

\begin{theorem}[\citealp{Myerson81}; Payment identity] \label{thm:Myerson}
An auction with an allocation rule $x$ and a payment rule $p$ is DSIC if and only if
\begin{enumerate}
    \item $x$ is a monotone allocation rule, and,
    \item bidder $i$ is charged $p_i(b_i, b_{-i}) = \critical_i(x, b_{-i})$.
\end{enumerate}
Further, for a given allocation rule $x$, the payment rule $p$ that satisfies the above two conditions is unique.\footnote{The version of the payment identity stated in \Cref{thm:Myerson} is not the most popular one. The more standard version states that the bidder $i$ must be charged a payment $p_i(b_i, b_{-i}) = \int_0^{b_i} zx'_i(z, b_{-i})\,dz$ for the auction $(x, p)$ to be DSIC. However, we find the variant in terms of the critical bids to be more useful for the discussions in the paper. For a detailed treatment of the two versions, we direct the readers to (for example) \citet{Hartline13}.}
\end{theorem}

\subsection{Revenue Maximization and Virtual Surplus}

\citet{Myerson81} discusses a recipe to optimize the expected revenue for a Bayesian auctioneer by a reduction to surplus optimization via virtual values.

\begin{definition}[Virtual Values]
For a continuous distribution with a cumulative distribution function $F$ and a density function $f$, the virtual value function $\vv(\cdot)$ is given by
$\vv(v) = v - \frac{1-F(v)}{f(v)}$.
\end{definition}

For an allocation $x = (x_1, \dots, x_n)$ for $n$ bidders with virtual values $\vv_1(v_1), \dots, \vv_n(v_n)$,
the virtual surplus of the allocation equals $\sum_{i = 1}^n \vv_i(v_i) \, x_i$.
\citet{Myerson81} shows that the expected revenue of a DSIC mechanism equals its expected virtual surplus.

\begin{theorem}[\citealp{Myerson81}] \label{thm:MyersonVV}  
Consider a DSIC mechanism (or an EPIC mechanism with a direct-revelation mechanism) $(x, p)$ where bidders' values are drawn from the product distribution $\Pi_{i = 1}^n D_i$.
The expected revenue from bidder $i$ conditioned on the bids $v_{-i}$ of bidders other than $i$ equals
$$\E_{v_i \sim D_i}[p_i(v_i, v_{-i})] = \E_{v_i \sim D_i}[\vv_i(v_i) \, x_i(v_i, v_{-i})].$$
Further, the total revenue from the mechanism $(x, p)$ equals
$$\E_{\vec{v} \sim \Pi_{i=1}^n D_i}[\sum_{i = 1}^n p_i(v_i, v_{-i})] = \E_{\vec{v} \sim \Pi_{i=1}^n D_i}[\sum_{i = 1}^n \vv_i(v_i) \, x_i(v_i, v_{-i})].$$
\end{theorem}

When the virtual value function of a distribution is monotone non-decreasing in value, the allocation rule that pointwise maximizes virtual value is monotone and, therefore, revenue-optimal alongside a payment rule that satisfies the payment identity (\Cref{thm:Myerson}).  
Such distributions with a monotone virtual value function are said to be \emph{regular}.  

When the virtual value $\vv$ is not monotone non-decreasing, \citet{Myerson81} describes an ironing procedure to obtain the \emph{ironed virtual value} $\overline{\vv}$ (for a detailed treatment of the ironing procedure, see for example, \citealp{Hartline13}).
The ironed virtual value $\overline{\vv}(v)$ is monotone non-decreasing in $v$. Further, as long as two values $v$ and $w$ with the same ironed virtual values get allocated identically, i.e, $x_i(v, v_{-i}) = x_i(w, v_{-i})$ whenever $\overline{\vv}_i(v) = \overline{\vv}_i(w)$, the expected revenue collected from bidder $i$ equals its expected ironed virtual surplus.
Thus, the allocation rule that pointwise maximizes the ironed virtual surplus optimizes the expected revenue when bidders' values are drawn from distributions that are irregular.

For downward-closed feasibility constraints, the ironed-virtual-surplus-optimal allocation will not allocate bidders with a negative ironed virtual value.
Let $r_i$ be the supremum over all values for which the ironed virtual value of the distribution $D_i$ is negative.
Then $r_i$ is the \emph{monopoly reserve} of the distribution $D_i$.
In order to b allocated, a bidder needs to have a value at least $r_i$, and thus, bidder $i$'s critical bid (and thereby, its payment upon getting allocated) is at least $r_i$.

\subsubsection{$\alpha$-Strongly Regular and MHR Distributions}

While regularity imposes that the virtual value function of a distribution is monotone non-decreasing, stronger monotonicity conditions can be imposed on the virtual value $\vv(\cdot)$.
A distribution is said to $\alpha$-strongly regular if for all $v \geq \hat{v}$, $\vv(v) - \vv(\hat{v}) \geq \alpha \, (v - \hat{v})$.
Note that $0$-strong regularity is identical to the regularity condition.
A distribution is said to satisfy the monotone hazard rate (MHR) condition if it is $1$-strongly regular.

\subsection{Commitment Schemes}

We conclude our preliminaries by briefly discussing our necessary cryptography primitives.

A commitment scheme $\comm$ allows credibly committing to a message without having to share the message publicly at the time of commitment.
It takes as input a message $m$ and a pad $r$ to output a commitment $\comm(m, r)$.
We provide an informal sketch of the properties desired from such a commitment scheme.

We say that the commitment scheme is perfectly binding if for a message $m$ and a pad $r$, $\comm(m, r) \neq \comm(\hat{m}, \hat{r})$ unless $m = \hat{m}$ and $r = \hat{r}$.
$\comm$ is perfectly hiding if for randomly drawn pads $r, \hat{r}$, the distribution of $\comm(m, r)$ and the distribution $\comm(\hat{m}, \hat{r})$ are indistinguishable for all messages $m$ and $\hat{m}$.
A perfectly hiding and binding commitment is ideal for a commit/reveal scheme.
A person can commit to a message $m$ by publishing $c = \comm(m, r)$.
No information about $m$ can be inferred just by viewing $c$, until the author of $m$ reveals the message.
Since the scheme is binding, the author cannot reveal a different message $m'$ with the same commitment $c$.

A commitment scheme is non-malleable if, at a high-level, an adversary cannot generate a commitment to any function $f(m)$ given the commitment $\comm(m, r)$ (unless, of course, $f$ is the identity function).
For a more detailed discussion on malleability, we refer the reader to \citet{DDN91} and \citet{FF00}.

While perfectly hiding and perfectly binding commitment schemes are impossible to construct, computationally binding and hiding schemes are sufficient for the discussions in the paper.
However, we assume that the commitment scheme is perfectly binding and perfectly hiding since it makes our arguments cleaner.

\section{The Deferred-Revelation Auction Over a Public Ledger} \label{sec:DRA}

In this section, we formally state the Deferred-Revelation Auction for matroid feasibility constraints.

For a perfectly hiding, binding and non-malleable commitment scheme $\comm$, an auctioneer optimizing its revenue and quasi-linear bidders, we describe the \emph{Deferred-Revelation Auction (DRA)} implemented over a public ledger.
We model the public ledger as an abstract functionality such that anything written on the ledger is observable to all entities participating in the auction.

At a high level, for a matroid feasibility constraint $\Feasibility$, the DRA is executed over three phases.
In the initialization phase, the auctioneer writes in the public ledger, a smart contract implementing the DRA.
The auctioneer is neither constrained to honestly implement the feasibility constraint $\Feasibility$ (however, the set of real bidders that are allocated at the end of the auction should always be feasible irrespective of the bids placed by the bidders) nor is it required to report the value distributions $D_1, \dots, D_n$ of the $n$ bidders truthfully to the smart contract.
However, we restrict the auctioneer to report only matroid feasibility constraints to the smart contract.
Bidders also submit commitments to their bids during the initialization phase.
We assume that the public ledger knows an upper bound on the maximum monopoly reserve amongst the distributions $D_1, \dots, D_n$, independent of the distributions reported by the auctioneer.
The bidders submit a collateral alongside their commitments, which can be a function of the upper bound on the largest monopoly reserve and the number of commitments written on the ledger during the initialization phase.

The revelation phase begins once the auctioneer signals the end of the initialization phase.
Bidders are expected to reveal their commitments during the revelation phase.
Since we assume a perfectly binding commitment scheme, we assume that either bidders reveal their bids honestly or they do not reveal their bids at all.
Bidders will not be able to find a random pad corresponding to a bid different to the one they committed to that also generates the same commitment.
Bids that are not revealed by the end of the revelation phase will be slashed their entire collateral.

Finally, the allocation phase implements the revenue-optimal auction given the set of bids that have been revealed.
Bidders transfer a payment to the auctioneer.
While bidders can choose to abort during the allocation phase to avoid making a payment (and not receive the good too), such behaviour can be disincentivized too by slashing their collateral.
Thus, we will assume bidders transfer the necessary payment to the auctioneer if allocated (or do not participate at all and not even submit a commitment during the initialization phase).

\begin{definition}[Deferred-Revelation Auction (DRA) over a Public Ledger] \label{def:DRA}
For a matroid feasibility constraint $\Feasibility = M(E, \I)$, the DRA is implemented over the following phases:
\begin{enumerate}
    \item (Announcement phase)
        \begin{itemize}
            \item The auctioneer announces the commencement of the auction.
            \item The auctioneer learns the (regular) distributions $D_1, \dots, D_n$ of the values of the $n$ bidders interested in participating in the auction.
            Let $\vv_1, \dots, \vv_n$ be the virtual value functions of the $n$ distributions.
        \end{itemize}
    \item (Initialization phase)
    \begin{itemize}
        \item Bidders $1 \leq i \leq \hat{n}$ confirm participation in the auction by writing their identifier $i$, their bid $b_i$ and $c_i = \comm(i, b_i, r_i)$, their commitment to bid $b_i$ on the ledger for a randomly drawn pad $r_i$.
        Note that the commitments can belong to both the $n$ interested bidders, or can be fabricated bids submitted by the auctioneer. Let $\vec{b}$ be the set of fabricated bids committed by the auctioneer.
        \item The ledger learns the distributions of the committed bidders.\footnote{As we will see in later sections, it is sufficient if the ledger learns an upper bound on the monopoly reserves of the value distributions and not necessarily the entire distributions.} As a function of the distributions and the number of commitments written on the ledger, the ledger announces a collateral $f$ to be collected from all committed bids.
        \item All bidders with committed bids can either choose to abort, or submit a collateral $f$ as demanded by the ledger.
        \item For the $\hat{n}$ bids written on the ledger, the auctioneer declares a matroid feasibility constraint $\hat{\Feasibility}$ and distributions $\hat{D}_1, \dots, \hat{D}_{\hat{n}}$ with virtual values $\hat{\vv}_1, \dots, \hat{\vv}_n$ and commits to implementing the revenue-optimal mechanism by writing a smart contract on the ledger.
        \item The auctioneer signals the end of the initialization phase, triggering the revelation phase.
    \end{itemize}
    \item (Revelation phase)
    \begin{itemize}
        \item Each bidder either writes their identifier $i$, bid $b_i$ and the random pad $r_i$ such that $c_i = \comm(i, b_i, r_i)$ on the ledger or remains silent and conceals its bid.
        \item The auctioneer triggers the end of the revelation phase.
        The collateral of all commitments without corresponding revealed bids in the revelation phase are burnt.
        \item The auctioneer calls the smart-contract that implements the revenue-optimal mechanism on the set of revealed bids.
    \end{itemize}
    \item (Allocation and payment phase)
    \begin{itemize}
        \item The smart contract chooses a subset $S \in \hat \Feasibility$ of bidders who revealed their bids with the largest virtual surplus, i.e, bidder $i$ is allocated if $i \in S$ that satisfies
        $\arg \max_{S \in \hat{\mathcal{\Feasibility}}} \{\sum_{i \in S} \hat{\vv}_i(b_i)\}$ (we assume ties are broken lexicographically).\footnote{Even though lexicographic tie-breaking is not necessary and any deterministic tie-breaking rule would suffice, it cleans up our arguments that follow.}
        \item Each allocated bidder is charged their critical bid, $\inf \{\hat{b}_i : S = \arg \max_{S \in \Feasibility} \sum_{j \in S, j \neq i} \hat{\vv}_j(b_j) + \hat{\vv}_i(\hat{b}_i)\}$, where the maximum is taken over all subsets $S$ of bidders with revealed bids.
        \item The collateral of allocated bidders that default on the payment demanded by the mechanism are burnt.
        Further, such bidders are removed from the set of allocated bidders.\footnote{We assume the smart contract ensures a fair exchange, and that the auctioneer does not default on allocating the bidders after receiving payments.}
    \end{itemize}
\end{enumerate}
\end{definition}

The auctioneer can be strategic in the DRA in exactly the following ways.
In the initialization phase, the auctioneer can fabricate any number of bids.
Further, the auctioneer can misreport its beliefs on the value distributions $D_1, \dots, D_n$.
Finally, the auctioneer can also write a smart contract optimizing revenue for some matroid feasibility constraint $\hat{\Feasibility}$ other than $\Feasibility$.
However, remember that the auctioneer's actions in the revelation phase must finally ensure that only a feasible set of real bidders are allocated by the mechanism irrespective of the bids placed by the bidders.
In the revelation phase, the auctioneer sees all the bids revealed by the bidders.
The auctioneer can choose to conceal some of its fabricated bids as a function of the bids revealed by bidders.
The mechanism does not require any input from the auctioneer or the bidders in the allocation and payment phase, and thus, the auctioneer cannot tamper with the mechanism once it triggers the end of the revelation phase.

We say that the DRA is \emph{credible} if, for the collateral charged by the mechanism, the auctioneer optimizes its expected revenue by (a) not fabricating any bids, (b) reporting the distributions $D_1, \dots, D_n$ truthfully to the mechanism and (c) reporting the feasibility constraint $\Feasibility$ truthfully to the mechanism.

Annotating the allocation and payment rules with the auctioneer's strategy will be quite useful.
Suppose $\strategy$ is an auctioneer's strategy that fabricates a set of bids $\vec{b}$.
We define $x(b_1, \dots, b_n, \vec{b} | \strategy)$ and $p(b_1, \dots, b_n, \vec{b} | \strategy)$ to be the allocation and payment rules induced by $\strategy$.
In particular, if the strategy $\strategy$ consists of revealing all of the fake bids $\vec{b}$, we denote the induced allocation and payment rules by $x(b_1, \dots, b_n, \vec{b})$ and $p(b_1, \dots, b_n, \vec{b})$.
For any strategy $\strategy$, the critical bid for bidder $i$ when the auctioneer reveals all its fake bids is given by $\critical_i$.
Finally, when the auctioneer truthfully informs the ledger of the distributions $D_1, \dots, D_n$ (and any arbitrary distributions for its fake bids), we denote the resulting allocation and payment rules by $x^{\OPT}(b_1, \dots, b_n, \vec{b} | \strategy)$ and $p^{\OPT}(b_1, \dots, b_n, \vec{b} | \strategy)$ respectively.

It is not hard to see that there exists an optimal strategy for the auctioneer that is deterministic.
Consider any strategy for the auctioneer that randomizes over a support of deterministic strategies.
Then, the expected revenue of the randomized strategy (expectation taken over both, the bids placed by bidders and the randomization over deterministic strategies) is at most that of the strategy with the highest expected revenue in the support (expectation over the bids submitted by the bidders).
Thus, for the rest of our discussions, we will assume that the auctioneer's strategy is deterministic.

Indeed, if the auctioneer is honest, we are implementing the revenue-optimal mechanism with DSIC payments.
Thus, bidders optimize their utilities by bidding truthfully.

\begin{theorem}
    For the DRA with distributions $D_1, \dots, D_n$ of bidder values and an arbitrary feasibility constraint $\Feasibility$, when the auctioneer is honest, bidders optimize their utilities by bidding truthfully, and the resulting equilibrium is revenue-optimal.
\end{theorem}

Conditioned on bidders placing truthful bids, we will prove that the auctioneer will also optimize its revenue by not fabricating any bids for some sufficiently large penalty $f$ for concealing bids.
The above would establish that truthful behaviour by all agents is a sub-game perfect Nash equilibrium and thus, the DRA is credible.

Since we fix bidders to bid truthfully, we will pretend that the penalty paid by the auctioneer for concealing bids is not burnt, but rather, is transferred to allocated bidders as a discount on their purchase.
Analyzing incentives from the auctioneer's viewpoint is identical to the DRA (the auctioneer's utility is the same as long as the auctioneer does not get the slashed collateral), but is more convenient since auctioneer's revenue (including the fines) equals the sum of payments made by the bidders.

\subsection{A Brief Aside on the Private Communication Model}

    While understanding the DRA implemented over a public ledger is the primary focus of the paper, we prove a few lower bounds when the auctioneer communicates with the bidders via private communication channels.
    We provide intuition on the private communication model necessary to parse our results, but encourage the reader to look at \citet{FW20} for a formal model.
    In the private communication model, the bidders learn about competing bids only when the auctioneer informs them about the other bids.

    The auctioneer's strategy is richer in the private communication model in the following ways.
    The auctioneer can commit to a different sets of bid to different bidders.
    The bidders cannot learn about the contradicting set of bids since the only way to communicate is via the auctioneer.
    Further, the auctioneer triggers the initialization, revelation and allocation phases separately to each bidder since there is no common announcement platform that can simultaneously inform all the bidders.
    This allows the auctioneer to trigger the revelation phase for some bidders, learn their bids and adapt its strategy accordingly before triggering the initialization/revelation phase for the others. 
    
    We assume that the auctioneer can provably burn the collateral of concealed bids in the private communication model.
    The above assumption is a deviation from the model described by \citet{FW20}, who assume that the collateral is split amongst all the allocated bidders.
    While a valid solution, the collateral required to make the DRA credible without a proof-of-burn could be much larger than the variant with proof-of-burn.    
    For example, if the auctioneer commits to a bid $b_{i^*}$ using the same identifier $i^*$ to more than one bidder and conceals it from everyone, the auctioneer can forward the proof-of-burn for the collateral $f$ submitted by the (fake) identity $i^*$ to all bidders.
    However, with no proof-of-burn, the only way for the auctioneer to prove that the collateral $f$ was disposed off is by transferring at least $f$ to each allocated bidder.
    We discuss the role of proof-of-burn in more detail in \Cref{sec:LowerPrivate}.


\section{Warm Up: Single Item for MHR Distributions} \label{sec:WarmUp}

\citet{FW20} and \citet{CFK24} show that a collateral equal to the largest monopoly reserve is sufficient for the DRA to be credible in the single-item environment when all value distributions are MHR.
We provide a reinterpretation of their result by considering a much simpler single-agent problem, which extends to matroid feasibility constraints in a clean way.

\begin{theorem}[\citealp{FW20, CFK24}] \label{thm:BabyFW}
    Suppose the value distributions $D_1, \dots, D_n$ are all MHR in a  single-item environment. Then, a collateral equal to the largest monopoly reserve amongst the value distributions is sufficient for the DRA implemented over a public ledger to be credible.
\end{theorem}

To prove \Cref{thm:BabyFW}, we fix the bids $v_{-i}$ of all other bidders and the fabricated bids $\vec{b}$ inserted by the auctioneer, which will induce a single-agent allocation and payment rule $x_i(\cdot, v_{-i}, \vec{b} | \strategy)$ and $p_i(\cdot, v_{-i}, \vec{b} | \strategy)$.
We will prove that the auctioneer optimizes its expected revenue by not concealing any of the fake bids.
Indeed, if the auctioneer always reveals all its fabricated bids, the auctioneer essentially implements a different DSIC mechanism.
The revenue of the other mechanism is only dominated by the optimal revenue.

We proceed with the proof by considering two cases --- when bidder $i$'s value $v_i$ is at least its critical bid $\critical_i$ and when $v_i < \critical_i$.
When $v_i \geq \critical_i$, concealing bids will not ``un-allocate'' bidder $i$ and will only result in a penalty to be paid by the auctioneer.
In \Cref{thm:BabyMoreThanCritical}, we will prove that concealing bids will not increase the auctioneer's expected revenue when $v_i < \critical_i$ either.

\begin{lemma} \label{thm:BabyBeatCritical}
    For a single-item environment with regular value distributions $D_1, \dots, D_n$, suppose the auctioneer plays a strategy $\strategy$ that involves inserting a set of fake bids $\vec{b}$.
    For a set of bids $v_{-i}$ of bidders apart from $i$,
    $$\E_{v_i \sim D_i}[p_i(v_i, v_{-i}, \vec{b} | \strategy) \times \mathbbm{1}(v_i \geq \critical_i)] \leq \E_{v_i \sim D_i}[\vv_i(v_i) \, x_i(v_i , v_{-i}, \vec{b} | \strategy) \times \mathbbm{1}(v_i \geq \critical_i)].$$
\end{lemma}
\begin{proof}
    When $v_i \geq \critical_i$, bidder $i$ wins the auction and is allocated the item.
    Concealing bids will only decrease bidder $i$'s critical bid, and hence, the payment made by $i$, not to mention the penalty to be paid by the auctioneer.
    Thus, the auctioneer's optimal strategy (optimal over all strategies that fabricate the bids $\vec{b}$) is to reveal all of its fabricated bids, which induces the allocation rule $x_i(\cdot, v_{-i}, \vec{b})$ with corresponding DSIC payments $p_i(\cdot, v_{-i}, \vec{b})$. Thus,
    \begin{equation}
        \notag
        \begin{split}
            \E_{v_i \sim D_i}[p_i(v_i, v_{-i}, \vec{b} | \strategy) \times \mathbbm{1}(v_i \geq \critical_i)] &\leq \E_{v_i \sim D_i}[p_i(v_i, v_{-i}, \vec{b}) \times \mathbbm{1}(v_i \geq \critical_i)] \\
            &= \E_{v_i \sim D_i}[\vv_i(v_i) \, x_i(v_i , v_{-i}, \vec{b}) \times \mathbbm{1}(v_i \geq \critical_i)].
        \end{split}
    \end{equation}
The equality follows from \Cref{thm:MyersonVV}.
When $v_i \geq \critical_i$, $x(v_i, v_{-i}, \vec{b}) = x(v_i, v_{-i}, \vec{b} | \strategy) = 1$, and therefore,
    \begin{equation}
        \notag
        \begin{split}
            \E_{v_i \sim D_i}[p_i(v_i, v_{-i}, \vec{b} | \strategy) \times \mathbbm{1}(v_i \geq \critical_i)] &\leq \E_{v_i \sim D_i}[\vv_i(v_i) \, x_i(v_i , v_{-i}, \vec{b} | \strategy) \times \mathbbm{1}(v_i \geq \critical_i)]. \qedhere
        \end{split}
    \end{equation}

\end{proof}

\begin{lemma}  \label{thm:BabyMoreThanCritical}
    In a single-item environment with MHR value distributions $D_1, \dots, D_n$, let $r_i$ be the monopoly reserve of distribution $D_i$.
    Consider any strategy $\strategy$ of the auctioneer that inserts a set of fake bids $\vec{b}$.
    Then, for a collateral $f = \max_i \{r_i\}$,
    $$\E_{v_i \sim D_i}[p_i(v_i , v_{-i}, \vec{b} | \strategy) \times \mathbbm{1}(v_i < \critical_i)] \leq \E_{v_i \sim D_i}[\vv_i(v_i) \, x_i(v_i , v_{-i}, \vec{b} | \strategy) \times \mathbbm{1}(v_i < \critical_i)].$$
\end{lemma}
\begin{proof}
    When $v_i < \critical_i$, the auctioneer has to conceal at least one of its fabricated bids to allocate bidder $i$.
    By including bidder $i$, the maximum payment the auctioneer can extract is $v_i$. Therefore,
    \begin{equation}
        \notag
        \begin{split}
            p_i(v_{i} , v_{-i}, \vec{b} | \strategy) \leq v_i \, x_i(v_i , v_{-i}, \vec{b} | \strategy) - r_i
            &\leq (v_i - r_i) \, x_i(v_i , v_{-i}, \vec{b} | \strategy)\\
            &\leq \vv_i(v_i) \, x_i(v_i , v_{-i}, \vec{b} | \strategy).
        \end{split}
    \end{equation}
    The second inequality follows since $x_i(v_i, v_{-i}, \vec{b} | \strategy) \in \{0, 1\}$. 
    The second line follows from a standard property of MHR distributions --- $(v_i - r_i) \leq \vv_i(v_i)$ for all $v_i \geq r_i$ (\Cref{thm:MHRLightTail} in \Cref{appendix:mhr-properties}).
    Taking expectation over all $v_i < \critical_i$ completes the proof.
\end{proof}

\Cref{thm:BabyFW} is a direct consequence of \Cref{thm:BabyBeatCritical} and \Cref{thm:BabyMoreThanCritical}.

\begin{proof}[Proof of \Cref{thm:BabyFW}]
    Combining \Cref{thm:BabyBeatCritical} and \Cref{thm:BabyMoreThanCritical}, we see that the auctioneer optimizes its virtual surplus from bidder $i$ by not concealing any of its fabricated bids $\vec{b}$.
    \begin{equation} \label{eqn:CombiningMoreandLessCrit}
        \begin{split}
            \E_{v_i \sim D_i}[p_i(v_i, v_{-i}, \vec{b} | \strategy)] &= \E_{v_i \sim D_i}[p_i(v_i, v_{-i}, \vec{b} | \strategy) \times \mathbbm{1}(v_i \geq \critical)] \\
            & \qquad + \E_{v_i \sim D_i}[p_i(v_i, v_{-i}, \vec{b} | \strategy) \times \mathbbm{1}(v_i < \critical)] \\
            &\leq \E_{v_i \sim D_i}[\vv_i(v_i) \, x_i(v_i, v_{-i}, \vec{b} | \strategy) \times \mathbbm{1}(v_i \geq \critical)] \\
            & \qquad + \E_{v_i \sim D_i}[\vv_i(v_i) \, x_i(v_i, v_{-i}, \vec{b} | \strategy) \times \mathbbm{1}(v_i < \critical)] \\
            &= \E_{v_i \sim D_i}[\vv_i(v_i) \, x_i(v_i , v_{-i}, \vec{b} | \strategy)].
        \end{split}
    \end{equation}

    We finish the proof by taking expectation over the values $v_{-i}$ of the other bidders and summing over all $i$.
    \begin{equation}
        \notag
        \begin{split}
            \E_{\vec{v} \sim \Pi_{i = 1}^n D_i}[\sum_{i = 1}^n p_i(v_i, v_{-i}, \vec{b} | \strategy)]
            &\leq \sum_{i = 1}^n \E_{v_{-i} \sim D_{-i}}\Big[\E_{v_i \sim D_i}[\vv_i(v_i) \, x_i(v_i , v_{-i}, \vec{b} | \strategy)]\Big] \\
            &\leq \sum_{i = 1}^n \E_{v_{-i} \sim D_{-i}}\Big[\E_{v_i \sim D_i}[\vv_i(v_i) \, x^{\OPT}_i(v_i , v_{-i})]\Big] \\
            &= \E_{\vec{v} \sim \Pi_{i = 1}^n D_i}[\sum_{i = 1}^n p^{\OPT}_i(v_i, v_{-i})].
        \end{split}
    \end{equation}
    The first line uses \Cref{eqn:CombiningMoreandLessCrit}, while the second observes that $x^{\OPT}$ pointwise optimizes virtual surplus.
    The last equality follows from \Cref{thm:MyersonVV}.
    Thus, the auctioneer optimizes revenue by reporting the distributions $D_1, \dots, D_n$ truthfully in the initialization phase and not inserting any fabricated bids.
\end{proof}

\begin{remark} \label{rem:PrivateBaby}
    Note that we did not explicitly use the public ledger assumption to prove \Cref{thm:BabyFW}.
    We only needed to ensure that (a) a bidder's critical bid (as communicated to the bidder by the auctioneer) does not increase on concealing bids, and (b) the penalty for concealing a bid from a bidder is at least the monopoly reserve of the bidder.
    Thus, a collateral equal to the monopoly reserve is sufficient in case of a single item in the private communication model.
    In \Cref{sec:LowerMHR}, we prove that the largest monopoly reserve is also necessary.
    Thus, there exists no separation between the required collateral that will make the DRA credible in the public and private communication models for a single-item.
\end{remark}

\section{Beyond Single Item and MHR Distributions} \label{sec:UpperBound}

In this section, we recreate the arguments in \Cref{sec:WarmUp} for matroid feasibility constraints.
The simplicity of the framework that analyzes one bidder at a time allows us to identify the salient properties of the feasibility constraint so that analogs of \Cref{thm:BabyBeatCritical} and \Cref{thm:BabyMoreThanCritical} continue to hold.

For \Cref{thm:BabyBeatCritical}, conditioned on the bids $v_{-i}$ and fabricated bids $\vec{b}$, we only require the critical bid of any agent to not increase when some subset of bids are concealed by the auctioneer (and remember that \Cref{thm:BabyBeatCritical} holds for all regular distributions, and not just MHR distributions).
\Cref{thm:BabyMoreThanCritical} requires that the distributions are MHR, and further that all allocated bidders with values smaller than their critical bids can \emph{each} be given a discount equal to the penalty $f$.
In other words, the number of allocated bidders with a value smaller than the critical bid is at least the number of bids concealed by the auctioneer.

In \Cref{sec:MatroidMHR}, we verify that matroid feasibility constraints indeed satisfy the above requirements.
We then relax the MHR requirement for \Cref{thm:BabyMoreThanCritical} in \Cref{sec:MatroidReg} to show that the DRA is credible for any $\alpha$-regular distribution for $\alpha > 0$ (however, we require a much larger collateral than when all value distributions are MHR).

\subsection{Matroid Feasibility Constraints and MHR Distributions} \label{sec:MatroidMHR}

\begin{theorem} \label{thm:MatroidMHR}
    For a matroid feasibility constraint $\Feasibility = M(E, \I)$ and MHR value distributions $D_1, \dots, D_n$ with corresponding monopoly reserves $r_1, \dots, r_n$, the DRA over a public ledger is credible for a collateral $f = \max_i \{r_i\}$.
\end{theorem}

We begin by proving that concealing bids does not increase the critical bids of any bidder in a matroid feasibility constraint.

\begin{lemma} \label{thm:MatroidConceal}
    Fix the bids $\vec{v}$ and fabricated bids $\vec{b}$ of bidders in a matroid feasibility constraint given by $M(E, \I)$, with corresponding virtual values $\hat{\vv}_i(\vec{v})$ and $\hat{\vv}_i(\vec{b})$ respectively.
    Let $W$ be the set of bidders allocated by the virtual-surplus-optimal mechanism $\hat{x}^{\OPT}$.
    Then, $\hat{x}^{\OPT}$ allocates all bidders in $W \setminus C$, irrespective of the set of fabricated bids $C$ concealed by the auctioneer. 
\end{lemma}

Remember that the virtual values in \Cref{thm:MatroidConceal} need not correspond to the distributions $D_1, \dots, D_n$.
It holds for any distributions and any corresponding virtual values reported by the auctioneer to the smart contract implementing the revenue-optimal mechanism for the reported distributions.

The following property of matroids will be useful for proving \Cref{thm:MatroidConceal}.

\begin{theorem}[\citealp{Brylawski73}; Symmetric basis-exchange property] \label{thm:MatroidBasisExchange}
    Let $M = (E, \I)$ be a matroid, and let $W, \hat{W} \in \I$ be two independent sets such that $|W| = |\hat{W}|$. Then, for any subset $D$ of $W$, there exists a subset $\hat{D}$ of $\hat{W}$ such that both $(W \setminus D) \cup \hat{D}$ and $(\hat{W} \setminus \hat{D}) \cup D$ are independent sets of $M$.
\end{theorem}

\begin{proof}[Proof of \Cref{thm:MatroidConceal}]
    Suppose for some set $C$ of concealed bids, the virtual-surplus-optimal mechanism allocates a set $\hat{W}$ of bidders such that $W \setminus C \not \subseteq \hat{W}$.
    Let $D = W \cap C$.
    We will apply the symmetric basis-exchange property of matroids (\Cref{thm:MatroidBasisExchange}) to the sets $W, \hat{W}$ and $D$.

    For convenience, add sufficiently many bidders with virtual value zero (and lowest in priority for tie-breaking) that can be added to any feasible set of cardinality strictly smaller than the rank of the matroid.
    After adding such bidders, we can assume that the cardinality of both $W$ and $\hat{W}$ both equal the rank of the matroid.
    Including these bids will not change the virtual surplus of either $W$ or $\hat{W}$.

    By the symmetric basis-exchange property, there exists a set $\hat{D} \subseteq \hat{W}$ such that $(W \setminus D) \cup \hat{D}$ and $(\hat{W} \setminus \hat{D}) \cup D$ both belong to the set of feasible allocations $\I$.
    Note that $(W \setminus D) \cup \hat{D}$ does not include any of the concealed bids $C$.
    Thus, if $(W \setminus D)$ has a virtual surplus greater than $(\hat{W} \setminus \hat{D})$, then $(W \setminus D) \cup \hat{D}$ has a virtual surplus greater than $\hat{W}$, the virtual-surplus-optimal allocation that does not include bids in $C$.\footnote{If $(W \setminus D)$ and $(\hat{W} \setminus \hat{D})$ have the same virtual surplus, and $(W \setminus D)$ has a higher priority, then $(W \setminus D) \cup \hat{D}$ will be chosen by the virtual-surplus-optimizer over $\hat{W}$, leading to a contradiction.}
    A contradiction. Similarly, if $(\hat{W} \setminus \hat{D})$ has a virtual value larger than $(W \setminus D)$, then $(\hat{W} \setminus \hat{D}) \cup D$ has a greater virtual surplus than $W$, a contradiction again.

    Thus, $\hat{W}$ must contain $W \setminus D$, which concludes the proof.    
\end{proof}

Now that we have proven bidders cannot be ``un-allocated'' by concealing bids, the proof of \Cref{thm:BabyBeatCritical} holds for matroid environments \emph{mutatis mutandis}.

\begin{lemma} \label{thm:MatroidBeatCritical}
    Let $M(E, \I)$ be a matroid feasibility constraint and let bidders' values be drawn from the regular distributions $D_1, \dots, D_n$.
    Let $\critical_i$ be bidder $i$'s critical bid conditioned on the bids $v_{-i}$ and fake bids $\vec{b}$.
    Then, for any strategy $\strategy$ of the auctioneer that injects the bids $\vec{b}$,
    $$\E_{v_i \sim D_i}[p_i(v_i, v_{-i}, \vec{b} | \strategy) \times \mathbbm{1}(v_i \geq \critical_i)] \leq \E_{v_i \sim D_i}[\vv_i(v_i) \, x_i(v_i , v_{-i}, \vec{b} | \strategy) \times \mathbbm{1}(v_i \geq \critical_i)].$$
\end{lemma}

Next, we prove an analog to \Cref{thm:BabyMoreThanCritical}.
Apart from the distributions $D_1, \dots, D_n$ being MHR, the key property that we used to prove \Cref{thm:BabyMoreThanCritical} was the following --- the auctioneer needs to conceal at least one bid in order to include a bidder with value below its critical bid.
This way, we were able to associate the penalty $f$ paid by the auctioneer as a discount offered to the allocated bidder.
In more general environments, suppose that the numbers of bids concealed by the auctioneer is at least the number of allocated bidders with bids smaller than their critical values.
Then, we will be able to once again associate the fine paid by the auctioneer as a discount $f$ per allocated bidder with $v_i < \critical_i$.
\Cref{thm:MatroidMoreThanCritical} will hold verbatim in such environments.
Indeed, the above condition is straightforward for matroids by \Cref{thm:Greedy}.
The auctioneer will have to conceal at least $k$ bids to allocate $k$ more bidders for a matroid feasibility constraint.

\begin{lemma} \label{thm:MatroidMoreThanCritical}
    Consider a matroid feasibility constraint $M(E, \I)$. Let the distributions $D_1, \dots, D_n$ of bidder values be MHR, and have monopoly reserves $r_1, \dots, r_n$ respectively.
    Then, for a penalty $f = \max_i \{r_i\}$, and for any strategy $\strategy$ of the auctioneer that injects the fake bids $\vec{b}$,
    $$\E_{v_i \sim D_i}[p_i(v_i , v_{-i}, \vec{b} | \strategy) \times \mathbbm{1}(v_i < \critical_i)] \leq \E_{v_i \sim D_i}[\vv_i(v_i) \, x_i(v_i , v_{-i}, \vec{b} | \strategy) \times \mathbbm{1}(v_i < \critical_i)].$$
\end{lemma}

We skip the proof of \Cref{thm:MatroidMHR} that proceeds identically to \Cref{thm:BabyFW}, combining \Cref{thm:MatroidBeatCritical} and \Cref{thm:MatroidMoreThanCritical}.

\begin{remark} \label{rem:MatroidPrivate}
Unlike the single-item environment, the largest monopoly reserve is not sufficient for general matroid environments under the private communication model.
\Cref{thm:MatroidMoreThanCritical} is no longer true under private communication.
To see the same, consider the environment with $2$ items and $2$ bidders.
Let the value distributions of both the bidders have a monopoly reserve equal to $1$.
Suppose the auctioneer informs both the bidders of two fabricated bids, one having a value $10$ and the other having a value $1 + \epsilon$, but does not inform the bidders of each other's bids.
If both the bidders have a value greater than $1+\epsilon$, the auctioneer can allocate both of them by hiding just one fabricated bid, the one with value $1+\epsilon$.
Thus, we cannot pretend to give a discount equal to $f$ to all allocated bidders whose values are smaller than their critical bids.

If, on the other hand, the collateral collected from each bid equals $k$ times the largest monopoly reserve (where $k$ is the rank of the matroid), then concealing a bid is similar to giving a discount $f$ to each of the $\leq k$ allocated bidders.
Thus, a collateral equal to $k$ times the largest monopoly reserve is sufficient for the DRA to be credible in the private communication model.
\end{remark}

\subsection{Matroid Feasibility Constraints and $\alpha$-strongly Regular Distributions} \label{sec:MatroidReg}

We show that the DRA is credible for a much larger collateral, for $\alpha$-strongly regular distributions.
We defer the proof to \Cref{appendix:proof-of-MatroidReg}.

\begin{theorem} \label{thm:MatroidReg}
    Consider a matroid feasibility constraint $\Feasibility = M(E, \I)$. Let $D_1, \dots, D_n$ be $\alpha$-strongly regular bidder value distributions with corresponding reserves $r_1, \dots, r_n$ ($0 < \alpha < 1$). Then, any collateral $f$ satisfying
    \begin{equation} \label{eqn:RegColIneq}
        \frac{1}{\alpha} \times \Bigg(\frac{\max_i \{r_i\}}{(1-\alpha) \, f + \alpha \max_i \{r_i\}}\Bigg)^{\frac{1}{1-\alpha}} \times \Big(\frac{f}{\max_i \{r_i\}} \Big) \leq \frac{1}{n}
    \end{equation}
    makes the DRA credible.
\end{theorem}

\begin{remark}
    For $\gamma = \frac{f}{\max_i \{r_i\}}$, where $r_1, \dots, r_n$ are the monopoly reserves of the value distributions $D_1, \dots, D_n$, the condition in \Cref{eqn:RegColIneq} reduces to
    \begin{equation}
    \notag
        \frac{1}{\alpha} \times \Bigg(\frac{1}{(1-\alpha) \, \gamma + \alpha}\Bigg)^{\frac{1}{1-\alpha}} \times \gamma \leq \frac{1}{n}.
    \end{equation}
    \citet{CFK24} verify that for $\gamma(n, \alpha) = \Big(\frac{n}{\alpha} \Big)^{\frac{1-\alpha}{\alpha}}(1-\alpha)^{-\frac{1}{\alpha}}$, collecting a collateral $f = \gamma(n, \alpha) \cdot \max_i \{r_i\}$ satisfies the inequality in \Cref{eqn:RegColIneq}.
    Therefore, a collateral equal to the largest monopoly reserve multiplied by some polynomial in $n$ (for a fixed $\alpha$) is sufficient for the DRA to be credible.
    However, choosing $f = \gamma(n, \alpha) \cdot \max_i \{r_i\}$ might be much larger than necessary.
    For example, as $\alpha$ tends to $1$, $\gamma(n , \alpha) \xrightarrow{} \infty$.
    However, from \Cref{eqn:RegColIneq}, it can be observed that the smallest $\gamma \geq 1$ required to satisfy the inequality is decreasing in $\alpha$.
    While we do not focus on obtaining tighter closed form expressions for a sufficient collateral $f$ for $\alpha$-strongly regular distributions, we emphasis that scaling the largest monopoly reserve by some polynomial in $n$ is sufficient.
\end{remark}

\begin{remark}
    For regular value distribution that are not $\alpha$-strongly regular for any $\alpha > 0$, \citet{CFK24} show that the DRA is not credible for any collateral $f$, even for the single-item environment with iid bidders.
    Consequently, our findings for matroid feasibility constraints cannot be extended beyond $\alpha$-strongly regular value distributions.
\end{remark}

\section{Lower Bounds} \label{sec:LowerBounds}

In this section, we prove lower bounds on the deferred-revelation auction for MHR distributions.
We argue that our results are tight in the following ways.
The DRA is not credible for all MHR distributions when (a) the feasibility constraint is not a matroid, irrespective of the deposits collected from the bidders, (b) the collateral is smaller than the largest monopoly reserve and finally, when (c) the auctioneer uses private channels to communicate with the bidders, but collects the same deposits as suggested for communication over a public ledger. 

\subsection{Credibility Beyond Matroids}
\label{sec:lowerbound-beyond-matroids}

In this section, we consider the DRA for feasibility constraints beyond matroids.
Irrespective of the feasibility constraint and the collateral collected, we argue that the auctioneer will be able to increase its revenue by reporting a non-matroid feasibility constraint $\hat{\Feasibility}$ to the DRA.

\begin{theorem} \label{thm:CreativeConstraint}
    When the auctioneer can report non-matroid feasibility constraints $\hat{\Feasibility}$ to the ledger,
    there exist MHR distributions for which the DRA is not credible irrespective of the (downward-closed) feasibility constraint $\Feasibility$ or the collateral $f$.
\end{theorem}
\begin{proof}
    Let the bidders' values all be drawn iid from the exponential distribution with an expected value $1$, a cumulative distribution function $1 - e^{-t}$ and a density $e^{-t}$.
The virtual value for a bidder with value $v$ is given by $\vv(v) = v-1$, and thus, the monopoly reserve equals $1$.

Consider the following deviation by the auctioneer.
The auctioneer fabricates the existence of two bidders $x$ and $y$, and reports the feasibility constraint $\hat{\Feasibility} = \{W \cup \{x\}, W : W \in \Feasibility \} \cup \{y\}$.
In other words, any feasible set $W$ under $\Feasibility$ remains feasible under the constraint $\hat{\Feasibility}$.
In addition, the bidder $x$ can be allocated alongside any feasible set $W$ without violating $\hat{\Feasibility}$.
However, if user $y$ is allocated, none of the remaining bidders can be allocated.

In the initialization phase, the auctioneer reports that the values of the bidders $x$ and $y$ are drawn from the exponential distribution with mean $1$, similar to the real bidders.
For a given collateral $f$, the auctioneer fabricates bids $b_x = b_y = f+2$ for both the bidders $x$ and $y$.

In the revelation phase, the auctioneer chooses a real bidder $i$ such that $\{i\}$ is feasible (such an $i$ exists since $\Feasibility$ is downward closed).
If the auctioneer observes a bidder $j \neq i$ with a value $v_j$ larger than $1$, the auctioneer reveals both the bids $b_x$ and $b_y$.
Note that fabricating and revealing the bids $x$ and $y$ changes neither the virtual-value-optimal allocation, nor the payments collected from the bidders.
The auctioneer does not pay any penalty since it reveals both the fabricated bids.

Suppose that all bidders apart from $i$ bid below the monopoly reserve $1$.
The auctioneer then deviates to the following strategy.
It reveals both $b_x$ and $b_y$ when bidder $i$ bids $v_i \leq f+2$.
As in the previous case, the allocation and payments remain unchanged and thus, the auctioneer's revenue does not decrease from being strategic.

However, whenever $v_i > f+2$, the auctioneer conceals $b_x$ and reveals $b_y$.
Since all other bidders bid below the monopoly reserve, and bid $x$ is concealed, the auction reduces to a second-price auction between the bidders $i$ and $y$.
Thus, bidder $i$ is charged $b_y = f+2$.
On the other hand, the auctioneer conceals a bid, and therefore pays a penalty $f$.
Its net revenue equals $2$.
However, when the auctioneer is truthful, bidder $i$ is allocated and is charged the monopoly reserve $1$, which is smaller than the revenue from the strategic deviation.
\end{proof}

While modelling the DRA beyond single-item settings, we assumed that the ledger does not know the feasibility constraint, and the auctioneer can be strategic in informing the ledger.
Another reasonable model could be where the ledger inherently knows the feasibility constraint, but does not know the number of bidders participating in the auction.
For example, in a blockchain with multi-dimensional resources, the ledger can ensure that the quantity of each resource consumed by the set of allocated transactions is at most the budget for the corresponding resource. However, the ledger might still not know the number of transactions beforehand and also cannot distinguish between real and fabricated transactions.

Even for the alternate model where the auctioneer cannot misreport feasibility constraints, we argue that the DRA need not be credible for all MHR distributions for any non-matroid feasibility constraint $\Feasibility$ irrespective of the collateral.

\begin{theorem} \label{thm:non-matroid-impossibility}
    For any fixed downward-closed non-matroid feasibility constraint $\Feasibility$, the DRA is not credible even with a single real bidder whose value is drawn from the exponential distribution with mean $1$ irrespective of the collateral $f$.
\end{theorem}

We defer the proof of \Cref{thm:non-matroid-impossibility} to \Cref{appendix:proof-of-CreativeConstraint}.

\subsection{A Tight Lower Bound for MHR Distributions} \label{sec:LowerMHR}

In this section, we prove that collecting the largest monopoly reserve as collateral is necessary for the DRA to be credible.
This complements \Cref{thm:MatroidMHR}, which proves that the largest monopoly reserve is also sufficient.

We first consider the single-item environment with a single participating bidder.
As we will see, the lower bound for the single-bidder environment immediately extends to single-item multi-agent environments.
The strategy for the single-item single-bidder environment is also quite useful in establishing lower bounds for iid environments too, albeit requiring a more careful analysis.
We prove that the collateral needs to be at least the monopoly reserve for the DRA to be credible for two iid environments, the $k$-item $k$-bidder setting and the single item $n$-bidder setting.
Thus, the necessary collateral is large in both cases, when there is high supply and when there is high competition.

\begin{lemma} \label{thm:11MHR}
    Let the bidder's valuation be drawn from the exponential distribution with expected value $1$ in a single-item single-bidder environment.
    Then, for any $\epsilon > 0$, the DRA is not credible for a collateral $f = 1 -\epsilon$.
\end{lemma}


\begin{proof}
    Suppose that the DRA collects a collateral $1-\epsilon$.
    Consider the auctioneer deviating by inserting a fake bid $1 + \delta$ (we will determine $\delta$ later) and concealing it if the bidder's value is between $1$ and $1+\delta$.
    We will show that the above deviation is profitable by explicitly computing the difference between the revenue from the above deviation and from being honest.

    When the bidder has a value at least $1$ (i.e, the monopoly reserve), it obtains a revenue $1$ when being strategic and when being honest.
    For the bidder's value greater than or equal to $1+\delta$ (which happens with a probability $e^{-(1+\delta)}$), the auctioneer receives an additional $\delta$ from being strategic.
    When the value $v_i \in [1, 1+\delta]$ (which happens with a probability $e^{-1} - e^{-(1+\delta)}$), the auctioneer conceals its fabricated bid and receives a penalty $1 - \epsilon$.

    Summarizing the above, the expected difference in revenue from deviating and being honest equals $\delta \, e^{-(1+\delta)} - (1-\epsilon) \, (e^{-1} - e^{-(1+\delta)})$
    which is positive if and only if the auctioneer's deviation is profitable.
    Rearranging terms in the inequality, the deviation is profitable exactly when
    $$\frac{\delta \, e^{-(1+\delta)}}{(e^{-1} - e^{-(1+\delta)})} > (1-\epsilon).$$
    As $\delta \xrightarrow{} 0$, the left hand side approaches $1$.
    Thus, there exists a sufficiently small $\delta$ such that the auctioneer's deviation is profitable.
\end{proof}

\begin{theorem} \label{thm:nnMHR}
    For a single-item environment with bidder values drawn from MHR values distributions $D_1, \dots, D_n$ with corresponding reserves $r_1, \dots, r_n$, a collateral of at least $\max_i \{r_i\}$ is necessary for the DRA to be credible.
\end{theorem}

Next, we focus on iid environments.
We show that restricting bidders to have values drawn from the same distribution (which is a reasonable assumption in many blockchain application with pseudonymous identities) does not allow a smaller collateral to be sufficient.
We will prove the above results in the $k$-item $k$-bidder environment and the single-item $n$-bidder environment.

\begin{theorem} \label{thm:kkMHR}
    When bidder values are drawn iid from the exponential distribution with expected value $1$, the DRA is not credible for any fine $f < 1$ in the $k$-item $k$-bidder environment.
\end{theorem}

\begin{theorem} \label{thm:1nMHR}
    A fine equal to the monopoly reserve $1$ is necessary for the DRA to be credible in the single-item iid environment where all bidder values are drawn from the exponential distribution with expected value $1$.
\end{theorem}

We defer the proofs in this section to \Cref{sec:proofofkkMHR}.

\subsection{Separating Public Ledger and Private Communication for MHR Distributions} \label{sec:LowerPrivate}

Observe that the lower bounds discussed in \Cref{sec:LowerMHR} hold even for the private communication model.
However, while the largest monopoly reserve is both necessary and sufficient with an access to a public ledger, we show that the lower bound is not tight with private communication and the DRA would require a collateral strictly greater than the monopoly reserve to be credible.
This separates the smallest sufficient collateral in the public and private communication models.
Recall that we consider the private communication model where the auctioneer can provably burn the collateral associated with concealed bids.
Therefore, the requirement for a larger collateral goes beyond just convincing bidders that the concealed bid was actually penalized.

\begin{theorem} \label{thm:PrivatePublicSep}
    Consider a $k$-agent $k$-item environment where the auctioneer communicates to bidders through private channels.
    When all values are drawn iid from the exponential distribution with expected value $1$, there exists a sufficiently large $k$ such that the DRA is not credible for a penalty equal to the monopoly reserve $1$.
\end{theorem}

We defer the proof to \Cref{appendix:proof-of-PrivatePublicSep}.

\section{Ascending Deferred Revelation Auction for Matroids} \label{sec:ADRA}

Even for a single item, \citet{FW20} and \citet{CFK24} show that the DRA is not credible when the bidders' value distributions are not regular.
\citet{EFW22} move beyond the DRA to propose an ascending variant of the deferred-revelation auction which is credible for all distributions.
In this section, we will generalize their mechanism for all matroid feasibility constraints.

\subsection{The Mechanism}

At a high level, the \emph{Ascending Deferred-Revelation Auction (ADRA)} for a single item proposed by \citet{EFW22} implements a ``fast-forward'' version of the ascending price auction.
Similar to the DRA, the ADRA asks bidders to commit to their bids before the auction starts.
However, unlike the ascending-price auction which gradually increases the price of the item until all but one bidders drop out of the auction, the ADRA is aggressive in increasing its prices.
For example, the ADRA can double the price of the item in each iteration.
Note that it is possible that \emph{all} active bidders drop out of the auction simultaneously when the price is updated from say, $p$ to $2p$.
In such scenarios, the ADRA uses the bidders' commitments to decide the winner.
The bidders are asked to reveal their bids $b_1, \dots, b_n$, and the ADRA simulates the ascending-price auction between $p$ and $2p$ with the revealed bids to determine the winner.

Unlike the DRA where all the collateral has to be collected at the start of the mechanism, the ADRA starts by collecting the largest monopoly reserve amongst all value distributions from each bidder, and asks bidders to add to the collateral to proceed to successive stages of the auction.
Intuitively, strategically aborting a bid becomes progressively costly as the auction proceeds, and thus, the auctioneer cannot wait for many rounds before deciding whether to abort any of its fabricated bids.

We extend the ADRA to matroid environments by modifying a surplus-optimal ascending-price auction due to \citet{BdVSV11}.
They consider a notion of matroid feasibility that is much more general than the one we consider.
We discuss a simpler, yet equivalent restatement of their mechanism tailored specifically for our setting.

\begin{algorithm}[Surplus-optimal ascending-price auction $\mechanism$ for matroids, \citealp{BdVSV11}] \label{alg:APA}
For a feasibility constraint $\Feasibility$ containing $n$ bidders:
    \begin{enumerate}
        \item Initialize level $\ell = 0$ and price $p_{\ell} = 0$.
        \item At level $\ell = 0$, all bidders are $\mathsf{competing}$.
        Throughout the auction, bidders are either $\mathsf{competing}$ or have $\mathsf{dropped}$ out from the auction.
        Similarly, bidders have either been $\mathsf{promised}$ an item, or not.
        Let $\competeBid_{\ell}$ and $\promiseBid_{\ell}$ respectively be the set of $\mathsf{competing}$ and $\mathsf{promised}$ bidders at level $\ell$.
        \item \label{Bul:APALoop} While there exists a $\mathsf{competing}$ bidder that has not been $\mathsf{promised}$ an item, i.e, $\competeBid_\ell \not \subseteq \promiseBid_{\ell}$:
        \begin{enumerate}
            \item Update level $\ell \xleftarrow{} \ell+1$.
            \item Increase the price $p_{\ell} \xleftarrow{} p_{\ell-1} + \varepsilon$ for some small $\varepsilon > 0$.
            \item Ask all $\mathsf{competing}$ bidders with a value smaller than $p_{\ell}$ to $\mathsf{drop}$ out from the auction.\footnote{If multiple bidders quit when increasing the price from $p_{\ell}$ to $p_{\ell + 1}$, check whether there exists some subset $S'$ of the bidders that quit such that $S' \cup S \in \Feasibility$ for all feasible allocations $S \subseteq \competeBid_{\ell} \cup \promiseBid_{\ell - 1}$. In other words, identify whether there exists some subset of the bidders that quit that can be allocated irrespective of the allocation to the remaining bidders.
            Allocate all bidders in $S'$ and charge them a price $p_{\ell - 1}$. Update their status to $\mathsf{promised}$.}
            \item Update the set $\competeBid_{\ell}$ of $\mathsf{competing}$ bidders.
            \item \label{Bul:ClinchCheck} For each $\mathsf{competing}$ bidder $i$ such that $S \cup \{i\} \in \Feasibility$ for all feasible sets $S \subseteq \competeBid_{\ell} \cup \promiseBid_{\ell-1}$:
            \begin{enumerate}
                \item Bidder $i$ can be allocated irrespective of how the remaining bidders are allocated.
                $\mathsf{Promise}$ to allocate bidder $i$ at a price $p_{\ell}$.
                \item Update the status of bidder $i$ to $\mathsf{promised}$.
            \end{enumerate}
            \item Update the set $\promiseBid_{\ell}$ of all the $\mathsf{promised}$ bidders.
        \end{enumerate}
        \item Allocate all $\mathsf{promised}$ bidders at the price they were $\mathsf{promised}$ the item.
    \end{enumerate}
\end{algorithm}

While the ascending-price auction $\mechanism$ discussed in \Cref{alg:APA} optimizes surplus, we apply the virtual-pricing transformation of ascending-price auctions (see, for example, \citealp{GH23}) to $\mechanism$ to obtain an ascending-price mechanism that maximizes the ironed virtual surplus.

\begin{definition}[Virtual-pricing transformation of $\mechanism$, \citealp{GH23}]
    For the surplus optimal mechanism $\mechanism$, the virtual-pricing transformation $\hat{\mechanism}$ implements $\mechanism$ in ironed virtual value space, i.e, whenever $\mechanism$ posts a price $p_{\ell}$, $\hat{\mechanism}$ posts a price
    $$\theta^i_{\ell} = \sup \{\theta: \overline{\vv}_i(\theta) \leq p_{\ell}\}$$
    to bidder $i$.
    The mechanism $\hat{\mechanism}$ is identical to $\mechanism$ except for the choice of prices posted to the bidders.
\end{definition}

\begin{theorem} \label{thm:APAOpt}
As $\varepsilon \xrightarrow{} 0$, bidding truthfully (i.e, $\mathsf{competing}$ until the price reaches its value and then $\mathsf{dropping}$ out) is EPIC in the virtual-pricing transformation $\hat{\mechanism}$.
Further, when all bidders participate truthfully,
$\hat{\mechanism}$ optimizes the ironed virtual surplus for all values $v_1, \dots, v_n$ of the $n$ bidders as $\varepsilon$ approaches $0$.
\end{theorem}

We will adapt the virtual-pricing transformation $\hat{\mechanism}$ into an ADRA.
Intuitively, the ADRA implements a ``stop-start'' version of $\hat{\mechanism}$.
In each iteration, the ADRA aggressively increases the virtual price $p^{\vv}_{\ell}$ according to some pricing function $\incr$.
If a set $Q$ of bidders quit when the price increases in iteration $\ell$, the ADRA asks the bidders to reveal their committed bids, and simulates $\hat{\mechanism}$ between the virtual prices $p^{\vv}_{\ell-1}$ and $p^{\vv}_{\ell}$.
The ADRA knows the bids of all the bidders in $Q$ and can simulate $\hat{\mechanism}$ assuming that the bidders would have bid truthfully.
Bidders not in $Q$ will not quit in the simulation too since they decided to remain active in the ADRA even at the virtual price $p^{\vv}_{\ell+1}$.

Note that once the ADRA knows the bids of the users, simulating the outcome of $\hat{\mechanism}$ for $\varepsilon \xrightarrow{} 0$ is quite straightforward.
We know that a bidder $i$ with a virtual value $\vv_i(v_i)$ will remain $\mathsf{active}$ until the virtual price reaches $\vv_i(v_i)$.
Thus, we can arrange bidders in $Q$ in ascending order of virtual values $\vv_t(v_t), \vv_{t+1}(v_{t+1}), \dots, \vv_{t+r}(v_{t+r})$ and directly increase the virtual price to $\vv_i(v_i)$ for $t \leq i \leq t+r$ before finally updating it to $p^{\vv}_{\ell +1}$.

As with the ADRA for single-items, we will collect an additional collateral from the bidders in each successive round.
For a virtual price $p^{\vv}_{\ell + 2}$ in level $\ell + 2$, we will require a total collateral locked-in by each bidder to be $\max_i \vv_i^{-1}(p^{\vv}_{\ell + 2})$ during level $\ell$.
The bidder loses all of its collateral if any evidence of suspicious behaviour is found (for example, the bidder $\mathsf{quit}$, but refused to reveal its bid, or the bidder did not quit as soon as the price surpassed the bid that it committed to), in which case, the bidder is said to have $\mathsf{aborted}$ its bid.

Whenever a bidder $\mathsf{aborts}$ at some level $\ell$, we restart the execution of the virtual-pricing transformation $\hat{\mechanism}$ from scratch until the virtual-price reaches $p^{\vv}_{\ell}$, this time without the bid the aborted bid.
As discussed earlier, the ADRA can use the bids that have already been revealed to determine when the bidders will $\mathsf{quit}$.
All bidders with a virtual value greater than $p^{\vv}_{\ell}$ will not $\mathsf{quit}$ in the fresh simulation of $\hat{\mechanism}$ and thus, the mechanism does not need to know their values for the simulation.

For convenience, we run a fresh simulation of $\hat{\mechanism}$ up to a virtual price $p^{\vv}_{\ell}$ even when bidders $\mathsf{quit}$ and reveal their bids in each level, instead of a ``stop-start'' implementation of $\hat{\mechanism}$.
Since $\hat{\mechanism}$ is deterministic, the states reached by the fresh start and the stop-start implementation will be identical.


\begin{definition}[Ascending Deferred-Revelation Auction (ADRA) for matroids] \label{def:ADRA}
    For a matroid feasibility constraint $M(E, \I)$ and a price update rule $\incr$ such that $\incr(p) > p$, the ADRA is implemented over the following phases:
    \begin{enumerate}
        \item (Announcement phase)
        \begin{enumerate}
            \item The auctioneer announces the commencement of the auction.
            \item The auctioneer learns the distributions $D_1. \dots, D_n$ of the values of the $n$ bidders interested in participating in the auction.
            Let $\overline{\vv}_1, \dots, \overline{\vv}_n$ be the respective ironed virtual value functions.
        \end{enumerate}
        \item (Initialization phase)
        \begin{enumerate}
            \item Bidders $1 \leq i \leq \hat{n}$ confirm participation in the auction by writing their identifier $i$, their bid $b_i$ and $c_i = \comm(i, b_i, r_i)$, their commitment to bid $b_i$ on the ledger for a randomly drawn pad $r_i$.
            Note that the commitments can belong to both, the $n$ interested bidders or can be fabricated bids submitted by the auctioneer.
            \item The ledger learns the distributions of the committed bidders.\footnote{Unlike the DRA where it was sufficient for the ledger to learn only an upper bound on the monopoly reserve of the distributions, we require the ledger to learn the entire distribution of values of the bidders for ADRA. However, the value distribution of the fabricated bids are adversarially determined. The auctioneer can choose distributions for its fake bids that would maximize its revenue from deviating.} The ledger collects a collateral $f_0$ equal to the largest monopoly reserve amongst the value distributions.
            \item All bidders with committed bids can either choose to $\mathsf{abort}$, or submit a collateral as demanded by the ledger and remain $\mathsf{active}$.
            \item For the $\hat{n}$ bids written on the ledger, the auctioneer declares an arbitrary matroid feasibility constraint $\hat{\Feasibility}$.
            \item The auctioneer signals the end of the initialization phase, triggering the ascending-price phase.
        \end{enumerate}
        \item (Ascending-price phase)
        \begin{enumerate}
            \item Initialize level $\ell$ to $0$.
            Further, initialize the virtual price $p^{\vv}_0 = 0$.
            \item Throughout the auction, bidders in the auction will have one of the three states --- $\mathsf{active}$, $\mathsf{quit}$ or $\mathsf{aborted}$.
            Separately, bidders might have been $\mathsf{promised}$ an item or not.
            Let the set of $\mathsf{active}$ bidders at level $\ell$ be $\activeBid_{\ell}$ and the set of promised bidders be $\promiseBid_{\ell}$.
            \item \label{bul:ADRALoop} While there exists an $\mathsf{active}$ bidder that has not been promised an item, i.e, $\activeBid_{\ell} \not \subseteq \promiseBid_{\ell}$:
            \begin{enumerate}
                \item Increment level $\ell \xleftarrow{} \ell + 1$.
                Update virtual price to $p^{\vv}_{\ell} \xleftarrow{} \incr(p^{\vv}_{\ell - 1})$.
                \item \label{bul:ADRAQuit} Ask all bidders with a bid $b_i$ such that $\overline{\vv}_i(b_i) < p^{\vv}_{\ell}$ to $\mathsf{quit}$ and reveal their bids by writing their identifier $i$, bid $b_i$ and the random pad $r_i$ such that $c_i = \comm(i, b_i, r_i)$ on the ledger.
                \item Increase the collateral collected from all $\mathsf{active}$ bidders to $f_{\ell} = \max_i \overline{\vv}^{-1}(\incr^2(p^{\vv}_{\ell}))$, so that when bidders reveal their bids in the next level, their bids will be smaller than the collateral.
                \item A bidder aborts at level $\ell$ if:
                \begin{enumerate}
                    \item Bidder $i$ quits at level $\ell$, but does not reveal its bid or writes a bid $b'_i$ and a random pad $r'_i$ such that $c_i \neq \comm(i, b'_i, r'_i)$.
                    \item Bidder $i$ reveals $b_i$, but the bidder should not have quit during level $\ell$, i.e, $b_i$ does not satisfy
                    $p^{\vv}_{\ell - 1} \leq \overline{\vv}_i(b_i) < p^{\vv}_{\ell}$.
                \end{enumerate}
                \item If there exists a bidder that $\mathsf{quit}$ or $\mathsf{aborted}$ during level $\ell$:
                \begin{enumerate}
                    \item Burn the collateral of all the bidders that $\mathsf{abort}$.
                    \item Simulate $\hat{\mechanism}$ without the $\mathsf{aborted}$ bidders up to a virtual-price $p^{\vv}_{\ell}$ using the revealed bids for bidders that have $\mathsf{quit}$.
                    $\mathsf{Active}$ bidders will not quit in the simulation of $\hat{\mechanism}$ that stops at $p^{\vv}_{\ell}$ and thus, the simulation does not need to know their bid.
                    \item Let $\promiseBid_{\ell}$ be the set of all bidders $\mathsf{promised}$ by the simulation of $\hat{\mechanism}$.
                    Note that $\mathsf{promised}$ bidders do not $\mathsf{quit}$ by default. They $\mathsf{quit}$ only as dictated by \Cref{bul:ADRAQuit}.
                \end{enumerate}
            \end{enumerate}
            \item \label{bul:ADRAFinal} If the set of all $\mathsf{active}$ bidders have been promised an item, i.e, $\activeBid_{\ell} \subseteq \promiseBid_{\ell}$:
            \begin{enumerate}
                \item Increase level $\ell \leftarrow \ell + 1$. This is the final level.
                \item Ask all $\mathsf{active}$ bidders to reveal their bids by writing their identifier $i$, bid $b_i$ and the random pad $r_i$ such that $c_i = \comm(i, b_i, r_i)$ on the ledger.
                \item A bidder aborts at level $\ell$ if:
                \begin{enumerate}
                    \item Bidder $i$ quits at level $\ell$, but does not reveal its bid or write a bid $b'_i$ and a random pad $r'_i$ such that $c_i \neq \comm(i, b'_i, r'_i)$.
                    \item Bidder $i$ reveals $b_i$, but the bidder should not have quit during level $\ell$, i.e, $b_i$ does not satisfy
                    $p^{\vv}_{\ell - 1} \leq \overline{\vv}_i(b_i)$.
                \end{enumerate}
                \item Burn the collateral of all $\mathsf{aborted}$ bidders.
                \item Simulate $\hat{\mechanism}$ without the $\mathsf{aborted}$ bidders up to completion using the revealed bids, allocate bidders and charge payments as dictated by $\hat{\mechanism}$.
            \end{enumerate}
        \end{enumerate}
    \end{enumerate}
\end{definition}


Similar to the DRA, the auctioneer can be strategic in the ADRA by declaring a matroid feasibility constraint $\hat{\Feasibility} \neq \Feasibility$, while the final allocation to the real bidders must be feasible, irrespective of $\hat{\Feasibility}$ and values $v_1, \dots, v_n$ of the real bidders.
Indeed, the auctioneer could also fabricate bids and strategically $\mathsf{abort}$ them based on the bids revealed by the bidders.
It can also be strategic in deciding the value distributions of its fabricated bids.
However, remember that the auctioneer cannot misreport the distributions $D_1, \dots, D_n$ of the values of the real bidders, like in the DRA.
We want to argue that the ADRA is credible i.e, the auctioneer optimizes its expected revenue by not fabricating any bids.

\subsection{Analysis}

Observe that the communication complexity of the ADRA for matroids is at most the communication complexity for the single-item environment.
The competition to get allocated is only lower in general matroid environments in comparison to the single item, and thus, the auction terminates at a lower level than the single-item setting, thereby requiring lower communication.
\citet{EFW22} show that there exists $\incr$ for which the ADRA required only a constant communication in the single-item environment.
For the same choice of $\incr$, the ADRA also requires only a constant communication for arbitrary matroid environments.

\begin{theorem}
    There exists $\incr$ such that the expected number of levels of the ADRA is $O(\log (\E[\max_{i \in [n]} v_i]))$.
\end{theorem}

When the auctioneer is honest, the ADRA is just a fancy implementation of the ascending-price auction in \Cref{alg:APA}, and is thus EPIC for bidders to bid truthfully.
By \Cref{thm:APAOpt}. the ADRA is also revenue-optimal when the auctioneer is honest.

\begin{theorem}
    For the ADRA with distributions $D_1, \dots, D_n$ of bidder values and a matroid feasibility constraint $\Feasibility$, bidding truthfully is EPIC for bidders conditioned on the auctioneer playing the honest strategy.
    The resulting equilibrium in the ADRA is revenue-optimal.
\end{theorem}

Finally, we show the credibility of ADRA for matroids. 
We defer the proof to \Cref{appendix:proof-of-ADRA}

\begin{theorem}\label{thm:ADRA}
    The ADRA is credible irrespective of the bidders' value distributions $D_1, \dots, D_n$ and the matroid feasibility constraint $\Feasibility$.
\end{theorem}

\noindent\textbf{Acknowledgment.} We thank Matt Weinberg for helpful discussions throughout the duration of this work.

\bibliographystyle{ACM-Reference-Format}
\bibliography{references}

\appendix

\section{Properties of $\alpha$-Strongly Regular and MHR Distributions}
\label{appendix:mhr-properties}

The following lemmas argue that if the value is large, then the virtual value must also be quite large for $\alpha$-strongly regular distributions.

\begin{lemma}[\citealp{HR09}] \label{thm:MHRLightTail}
    Let $D$ be an MHR distribution with a virtual value function $\vv$ and a monopoly reserve $r$. Then, for all values $v \geq r$, $\vv(v) \geq v - r$.
\end{lemma}

\begin{lemma}[\citealp{FW20}] \label{thm:RegLightTail}
    Let $D$ be an $\alpha$-strongly regular distribution with a virtual value function $\vv$ and a monopoly reserve $r$. Then, for all values $v \geq r$, $\frac{1}{\alpha}\vv(v) \geq v - r$.
\end{lemma}

\noindent Note that \Cref{thm:MHRLightTail} is a special case of \Cref{thm:RegLightTail} with $\alpha = 1$.

\Cref{thm:MHRLightTail} is also tight.
The exponential distribution is a distribution that is MHR, but not $\alpha$-strongly regular for any $\alpha < 1$.
For an expected value $\mu$, the exponential distribution has a cumulative distribution function $1 - e^{-(\frac{v}{\mu})}$ and a density $\frac{1}{\mu}e^{-(\frac{v}{\mu})}$ for all $v \geq 0$.
This corresponds to a virtual value $\phi(v) = v - \mu$ and a monopoly reserve $\mu$.

The following lemma shows that the right-tail of $\alpha$-strongly regular distributions have to be light.

\begin{lemma}[\citealp{FW20}] \label{thm:MHRWeirdBound}
    Let $D$ be an $\alpha$-strongly regular distribution with a monopoly reserve $r$.
    Then, for all $f \geq r$,
    \begin{equation}
        \notag
        Pr_{v \sim D}(v \geq f) \leq Pr_{v \sim D}(v \geq r) \times \Bigg(\frac{r}{(1-\alpha) f + \alpha r}\Bigg)^{\frac{1}{1-\alpha}}.
    \end{equation}
\end{lemma}

\section{Omitted Proofs} \label{sec:OmittedProofs}

\subsection{Proof of \Cref{thm:MatroidReg}}
\label{appendix:proof-of-MatroidReg}

We aim to recreate the framework from the MHR scenario.
While \Cref{thm:MatroidBeatCritical} holds seamlessly for all regular distributions, \Cref{thm:MatroidMoreThanCritical} fails to hold beyond MHR value distributions.
We obtain a variant for more general distributions.
We identify two simple conditions for which the auctioneer cannot increase its revenue by concealing bids.
We then prove a variant of \Cref{thm:MatroidMoreThanCritical} for bids that do not satisfy these simple conditions.

The first of the two simple conditions is when the critical bid of a bidder is smaller than $f$ (the second condition is discussed later in the section). Concealing bids will not increase the revenue extracted from the bidder since any additional revenue obtained would be offset through the penalty $f$ paid to the bidder.
Thus, for most parts of the proof, we focus solely on the scenario when $\critical_i \geq f$.
We start by deriving analogs of \Cref{thm:MatroidBeatCritical} and \Cref{thm:MatroidMoreThanCritical}.

\begin{lemma} \label{thm:RegBeatCritical}
    For a matroid feasibility constraint with regular value distributions $D_1, \dots, D_n$,
    $$\E_{v_i \sim D_i}[p_i(v_i, v_{-i}, \vec{b} | \strategy) \times \mathbbm{1}(v_i \geq \critical_i, \critical_i \geq f)] \leq \E_{v_i \sim D_i}[\vv_i(v_i) \, x_i(v_i , v_{-i}, \vec{b} | \strategy) \times \mathbbm{1}(v_i \geq \critical_i, \critical_i \geq f)]$$
    for any strategy $\strategy$ of the auctioneer that fabricates the set of bids $\vec{b}$.
\end{lemma}
For a given set of bids $v_{-i}$ and $\vec{b}$, $\mathbbm{1}(\critical_i \geq f)$ is a non-negative constant that is independent of $v_i$.
Thus, \Cref{thm:RegBeatCritical} follows by multiplying both sides of \Cref{thm:MatroidBeatCritical} by $\mathbbm{1}(\critical_i \geq f)$.

\begin{lemma} \label{thm:MatroidRegIntermediate}
    For a matroid feasibility constraint with $\alpha$-strongly regular bidder value distributions ($0 < \alpha < 1$), bids $v_{-i}$ of bidders apart from $i$ and any strategy $\strategy$ of the auctioneer than fabricates a set of bids $\vec{b}$,
    \begin{equation}
        \notag
        \begin{split}
            \E_{v_i \sim D_i}[p_i(v_i, v_{-i}, \vec{b} | \strategy) \mathbbm{1}(v_i < \beta_i, f \leq \critical_i)] 
            &\leq \frac{1}{\alpha} \cdot \E_{v_i \sim D_i}[\vv_i(v_i) \, x_i(v_i, v_{-i}, \vec{b} | \strategy) \times \mathbbm{1}(f \leq v_i, f \leq \critical_i)] \\
            &\qquad \qquad - \E_{v_i \sim D_i}[\vv_i(v_i) \, x_i(v_i, v_{-i}, \vec{b} | \strategy) \times \mathbbm{1}(\beta_i \leq v_i, f \leq \critical_i)].
        \end{split}
    \end{equation}    
\end{lemma}
\begin{proof}
The following chain of inequalities are similar to \Cref{thm:MatroidMoreThanCritical}.
    \begin{equation}
        \notag
        \begin{split}
            \E_{v_i \sim D_i}[p_i(v_i, v_{-i}, \vec{b} | \strategy) \mathbbm{1}(v_i < \beta_i, f \leq \critical_i)] &\leq \E_{v_i \sim D_i}[(v_i - f) \, x_i(v_i, v_{-i}, \vec{b} | \strategy) \times \mathbbm{1}(v_i < \beta_i, f \leq \critical_i)] \\
            &\leq \E_{v_i \sim D_i}[(\frac{1}{\alpha}\vv_i(v_i) + r_i - f) \, x_i(v_i, v_{-i}, \vec{b} | \strategy) \times \mathbbm{1}(f \leq v_i < \beta_i)]
        \end{split}
    \end{equation}
The first inequality follows since the maximum payment extracted from a bidder is their value.
The second inequality follows by first ignoring all values $v_i < f$, in which case $v_i - f$ is negative, and then applying \Cref{thm:RegLightTail}, which argues $(v_i - r_i) \leq \frac{1}{\alpha} \vv_i(v_i)$ for all $v_i \geq r_i$ for all $\alpha$-strongly regular distributions.
We will be choosing $f \geq r_i$ for all bidders $i$, and thus,
    \begin{equation}
        \notag
        \begin{split}
            \E_{v_i \sim D_i}[p_i(v_i, v_{-i}, \vec{b} | \strategy) \mathbbm{1}(v_i < \beta_i, f \leq \critical_i)]
            &\leq \E_{v_i \sim D_i}[\frac{1}{\alpha}\vv_i(v_i) \, x_i(v_i, v_{-i}, \vec{b} | \strategy) \times \mathbbm{1}(f \leq v_i < \beta_i)] \\
            &= \frac{1}{\alpha} \cdot \E_{v_i \sim D_i}[\vv_i(v_i) \, x_i(v_i, v_{-i}, \vec{b} | \strategy) \times \mathbbm{1}(f \leq v_i, f \leq \critical_i)] \\
            & \qquad - \frac{1}{\alpha} \cdot \E_{v_i \sim D_i}[\vv_i(v_i) \, x_i(v_i, v_{-i}, \vec{b} | \strategy) \times \mathbbm{1}(\beta_i \leq v_i, f \leq \critical_i)] \\
            &\leq \frac{1}{\alpha} \cdot \E_{v_i \sim D_i}[\vv_i(v_i) \, x_i(v_i, v_{-i}, \vec{b} | \strategy) \times \mathbbm{1}(f \leq v_i, f \leq \critical_i)] \\
            & \qquad - \E_{v_i \sim D_i}[\vv_i(v_i) \, x_i(v_i, v_{-i}, \vec{b} | \strategy) \times \mathbbm{1}(\beta_i \leq v_i, f \leq \critical_i)].
        \end{split}
    \end{equation}
The final inequality follows since $\alpha \in (0, 1)$.
\end{proof}

We will now describe the second condition.
Let $\FullRank_i$ be the event that in the virtual-surplus-optimal allocation that does not include bidder $i$, replacing any allocated fabricated bid by bidder $i$ makes the allocation infeasible.
In other words, if bidder $i$ is to be included, it must be by displacing another real bidder.
Note that $\FullRank_i$ is a property of bids $v_{-i}$ and $\vec{b}$, and is independent of $v_i$.
In order to displace a real bidder, bidder $i$ will have to place a bid with a virtual value larger than that of the replaced bidder and thus, its critical bid $\critical_i$ is independent of $\vec{b}$.
Concealing bids will, therefore, not alter the payment made by bidder $i$.
Similar to the scenario $f \geq \critical_i$, $\FullRank_i$ is the easier case to analyze.
Hence, we will focus on $\EmptyRank_i$, the negation of $\FullRank_i$.
We do so by comparing the auctioneer's revenue from being strategic against the revenue in a digital goods environment (where any subset of bidders can be feasibly allocated) with an additional constraint that only bidders with a value larger than $f$ can be allocated.

\begin{lemma} \label{thm:RegEmptyRank}
    Consider the $\alpha$-strongly regular values distributions $D_1, \dots, D_n$ in a matroid environment, with reserves $r_1, \dots, r_n$ respectively.
    Let the collateral $f$ satisfy \Cref{eqn:RegColIneq}.
    Then,
\begin{equation}
    \notag
    \begin{split}
        \sum_{i = 1}^n \E_{\vec{v} \sim \Pi_{j = 1}^n D_j} \Big[\sum_{i=1}^n p_i(v_i, \vec{v}, \vec{b} | \strategy) \times \mathbbm{1}(f \leq \critical_i, \EmptyRank_i) \Big]
        &= \frac{1}{n} \sum_{i = 1}^n \E_{\vec{v} \sim \Pi_{j = 1}^n D_j} \Big[\vv_i(v_i) \times \mathbbm{1}(r_i \leq v_i, f \leq \critical_i, \EmptyRank_i) \Big].
    \end{split}
\end{equation}
\end{lemma}
\begin{proof}
Since the environment is given by a matroid feasibility constraint, an injection can be drawn from the set of allocated bidders with values smaller than their critical value $\critical_i$ and the set of fake bids concealed by the auctioneer.
We can therefore continue assuming that each bidder allocated because of concealed bids are given a discount $f$ each by the auctioneer.

We will combine \Cref{thm:RegBeatCritical} and \Cref{thm:MatroidRegIntermediate}.
However, before taking expectation with respect to 
$v_{-i}$, we will multiply both sides by $\mathbbm{1}(\EmptyRank_i)$.
\begin{equation}
    \notag
    \begin{split}
                    \E_{\vec{v} \sim \Pi_{j = 1}^n D_j} \Big[\sum_{i=1}^n p_i(v_i, \vec{v}, \vec{b} | \strategy) \times \mathbbm{1}(f \leq &\critical_i, \EmptyRank_i)\Big] \\ &\leq \frac{1}{\alpha} \cdot \sum_{i = 1}^n \E_{\vec{v} \sim \Pi_{j = 1}^n D_j}[\vv_i(v_i) \, x_i(v_i, v_{-i}, \vec{b} | \strategy) \times \mathbbm{1}(f \leq v_i, f \leq \critical_i, \EmptyRank_i)]
    \end{split}
\end{equation}
The right hand side can be upper-bounded by the auctioneer's revenue in the digital goods environment, where the bidders are allocated whenever their values are at least the collateral $f$.
By \Cref{thm:MyersonVV}, $\E_{\vec{v} \sim \Pi_{j = 1}^n D_j}[\vv_i(v_i) \times \mathbbm{1}(f \leq v_i, f \leq \critical_i, \EmptyRank_i)] = f \cdot Pr(v_i \geq f) \times \mathbbm{1}(f \leq \critical_i, \EmptyRank_i)$.
We therefore have
\begin{equation}
    \notag
    \begin{split}
        \E&_{\vec{v} \sim \Pi_{j = 1}^n D_j}\Big[\sum_{i=1}^n p_i(v_i, \vec{v}, \vec{b} | \strategy) \times \mathbbm{1}(f \leq \critical_i, \EmptyRank_i) \Big] \leq \frac{1}{\alpha} \cdot \sum_{i = 1}^n \E_{v_{-i} \sim D_{-i}} \Big[f \cdot Pr(v_i \geq f) \times \mathbbm{1}(f \leq \critical_i, \EmptyRank_i)\Big] \\
        &\leq \frac{1}{\alpha} \cdot \sum_{i = 1}^n \Bigg(\frac{r_i}{(1-\alpha) \, f + \alpha r_i}\Bigg)^{\frac{1}{1-\alpha}} \times \Big(\frac{f}{r_i} \Big) \cdot r_i Pr_{v_i \sim D_i}(v_i \geq r_i) \cdot Pr_{v_{-i} \sim D_{-i}}(f \leq \critical_i, \EmptyRank_i) \\
        &\leq \frac{1}{\alpha} \Bigg(\frac{\max_i \{r_i\}}{(1-\alpha) \, f + \alpha \max_i \{r_i\}}\Bigg)^{\frac{1}{1-\alpha}} \times \Big(\frac{f}{\max_i \{r_i\}} \Big) \sum_{i = 1}^n r_i \cdot Pr_{v_i \sim D_i}(v_i \geq r_i) \cdot Pr_{v_{-i} \sim D_{-i}}(f \leq \critical_i, \EmptyRank_i)
    \end{split}
\end{equation}
The second line is a direct consequence of \Cref{thm:MHRWeirdBound} in \Cref{appendix:mhr-properties}.
As for the final inequality, we prove that $\big(\frac{r}{(1-\alpha) \, f + \alpha r}\big)^{\frac{1}{1-\alpha}} \times \big(\frac{f}{r} \big)$ is monotone increasing in $r$ as long as $r \leq f$ in \Cref{thm:Mon}.

\begin{lemma} \label{thm:Mon}
    When $\frac{f}{t} \geq 1$, $\Bigg(\frac{t}{(1-\alpha) \, f + \alpha t}\Bigg)^{\frac{1}{1-\alpha}} \times \Big(\frac{f}{t} \Big)$ is increasing in $t$.
\end{lemma}
\begin{proof}
    Let $x = \frac{f}{t}$. We would want to prove $\Bigg(\frac{1}{(1-\alpha) \, x + \alpha}\Bigg)^{\frac{1}{1-\alpha}} \times x$ is decreasing in $x$ when $x \geq \frac{1}{\alpha}$.
    The above is equivalent to proving
    $$\frac{1}{(1-\alpha)x^\alpha + \alpha x^{-(1-\alpha)}}$$
    is decreasing when $x \geq 1$.
    The derivative of the denominator easily confirms the same.
\end{proof}

Since we choose $f$ such that $\frac{1}{\alpha} \times \big(\frac{\max_i \{r_i\}}{(1-\alpha) \, f + \alpha \max_i \{r_i\}}\big)^{\frac{1}{1-\alpha}} \times \big(\frac{f}{\max_i \{r_i\}} \big) \leq \frac{1}{n}$,
\begin{equation}
    \notag
    \begin{split}
        \E_{\vec{v} \sim \Pi_{i = 1}^n} \Big[\sum_{i=1}^n p_i(v_i, \vec{v}, \vec{b} | \strategy) \times \mathbbm{1}(f \leq \critical_i, \EmptyRank_i) \Big] &\leq \frac{1}{n} \sum_{i = 1}^n r_i \cdot Pr_{v_i \sim D_i}(v_i \geq r_i) \cdot Pr_{v_{-i} \sim D_{-i}}(f \leq \critical_i, \EmptyRank_i) \\
        &= \frac{1}{n} \sum_{i = 1}^n \E_{\vec{v} \sim \Pi_{j = 1}^n} \Big[\vv_i(v_i) \times \mathbbm{1}(r_i \leq v_i, f \leq \critical_i, \EmptyRank_i) \Big]
    \end{split}
\end{equation}
\end{proof}

\begin{proof}[Proof of \Cref{thm:MatroidReg}]
We recall our observations that the auctioneer cannot increase the revenue extracted from bidder $i$ whenever $f > \critical_i$ or the event $\FullRank_i$ happens.
After combining the above observation with \Cref{thm:RegEmptyRank}, we sketch a mechanism with a virtual surplus greater than the auctioneer's expected revenue from being strategic in the DRA.
The optimal mechanism will only obtain a greater virtual surplus and thus, a greater revenue.

Condition on $\critical_i < f$ or the event $\FullRank_i$ getting realized.
Then, the auctioneer cannot increase its revenue from bidder $i$ by concealing any bids.
Thus,
\begin{equation} \label{eqn:Smallf}
    \begin{split}
        \E_{\vec{v} \sim \Pi_{j = 1}^n}\Big[p_i(v_i, \vec{v}, \vec{b} | \strategy) \times \mathbbm{1}(f > \critical_i)\Big] &\leq \E_{\vec{v} \sim \Pi_{j = 1}^n}\Big[p_i(v_i, \vec{v}, \vec{b}) \times \mathbbm{1}(f > \critical_i)\Big] \\
        &= \E_{\vec{v} \sim \Pi_{j = 1}^n}\Big[\vv_i(v_i) \, x_i(v_i, \vec{v}, \vec{b}) \times \mathbbm{1}(f > \critical_i)\Big]
    \end{split}
\end{equation}

\begin{equation} \label{eqn:Largef}
    \begin{split}
        \E_{\vec{v} \sim \Pi_{j = 1}^n}\Big[p_i(v_i, \vec{v}, \vec{b} | \strategy) \times \mathbbm{1}(f \leq \critical_i, \FullRank_i)\Big] &\leq \E_{\vec{v} \sim \Pi_{j = 1}^n}\Big[p_i(v_i, \vec{v}, \vec{b}) \times \mathbbm{1}(f \leq \critical_i, \FullRank_i)\Big] \\
        &= \E_{\vec{v} \sim \Pi_{j = 1}^n}\Big[\vv_i(v_i) \, x_i(v_i, \vec{v}, \vec{b}) \times \mathbbm{1}(f \leq \critical_i, \FullRank_i)\Big]
    \end{split}
\end{equation}
Adding \Cref{eqn:Smallf} and \Cref{eqn:Largef}, summing over all bidders $i$ and adding \Cref{thm:RegEmptyRank}, we get
\begin{equation}
    \notag
    \begin{split}
        \sum_{i=1}^n \E_{\vec{v} \sim \Pi_{j = 1}^n}\Big[p_i(v_i, \vec{v}, \vec{b} | \strategy)\Big] &\leq \frac{1}{n} \sum_{i = 1}^n \E_{\vec{v} \sim \Pi_{j = 1}^n} \Big[\vv_i(v_i) \times \mathbbm{1}(r_i \leq v_i, f \leq \critical_i, \EmptyRank_i) \Big] \\
        &\qquad + \E_{\vec{v} \sim \Pi_{j = 1}^n}\Big[\vv_i(v_i) \, x_i(v_i, \vec{v}, \vec{b}) \times (1-\mathbbm{1}(f \leq \critical_i, \EmptyRank_i))\Big]
    \end{split}
\end{equation}

The following allocation rule achieves a virtual surplus at least the right hand side of the above inequality.
Compute the allocation rule $x(\vec{v}, \vec{b})$ from fabricating a set of bids $\vec{b}$.
\begin{enumerate}
    \item Amongst all bidders $i$ such that $v_{-i}$ satisfy $f \leq \critical_i$ and $\EmptyRank_i$, choose one uniformly at random and allocate to the chosen bidder $i$ if $v_i \geq r_i$.
    \item Allocate all other bidders according to $x(\vec{v}, \vec{b})$.
\end{enumerate}
We still have to verify that the above allocation rule is feasible.
We prove feasibility from allocating bidders according to bullet 1 by considering the following two cases -- whether bidder $i$ is allocated by $x(\vec{v}, \vec{b})$ or not.

When bidder $i$ is excluded by $x(\vec{v}, \vec{b})$, from the definition of $\EmptyRank_i$, there exists an allocated fabricated bid that can be feasibly replaced by bidder $i$ and thus, allocating bidder $i$ alongside the other allocated bidders is feasible.
If $x(\vec{v}, \vec{b})$ allocates bidder $i$, then, by \Cref{thm:MatroidConceal}, the optimal allocation with bidder $i$'s bid concealed allocates the same set of bidders and possibly one additional bidder.
From the definition of $\EmptyRank_i$, at least one of the allocated bids that can be feasibly replaced by $i$ must be fake.
Thus, including bidder $i$ does not break feasibility.
The optimal mechanism $x^{\OPT}$ only has a larger virtual surplus, and thus, a larger revenue that the auctioneer's revenue from being strategic.
\end{proof}

\subsection{Omitted Proofs in \Cref{sec:lowerbound-beyond-matroids}}
\label{appendix:proof-of-CreativeConstraint}

For \Cref{thm:non-matroid-impossibility}, we will recreate the proof of \Cref{thm:CreativeConstraint}.
For some bidder $i$ and a non-matroid feasibility constraint $\Feasibility$, we will find two disjoint sets of bidders $X$ and $Y$ of the same cardinality such that $X \cup \{i\}, Y \in \Feasibility$ and all feasible subsets of $X \cup Y \cup \{i\}$ other than $X \cup i$ have a cardinality at most $|X|$. 
Once such sets $X$ and $Y$ are identified, the remainder of the proof follows closely to \Cref{thm:CreativeConstraint}, replacing the bidders $x$ and $y$ with the sets $X$ and $Y$ respectively.

\begin{proof}[Proof of \Cref{thm:non-matroid-impossibility}]
We will first prove \Cref{thm:non-matroid-impossibility} assuming the existence of a bidder $i$, sets $X$ and $Y$ such that $|X| = |Y|$, $X \cup \{i\}, Y \in \Feasibility$ and all feasible subsets of $X \cup Y \cup \{i\}$ apart from $X \cup \{i\}$ have a cardinality at most $|X|$.
We show that such sets exist later in \Cref{lem:non-matroid}.

When there is only a single bidder $i$ whose value is drawn from the exponential distribution with mean $1$, the miner fabricates the bidders $X \cup Y$ and reports their distributions to be the same exponential distribution as $i$.
Further, it reports a bid $b_x = b_y = f+2$ for all bidders $x \in X$ and $y \in Y$ for any collateral $f$ charged by the DRA.

In the revelation phase, the auctioneer follows a strategy similar to \Cref{thm:CreativeConstraint}.
If bidder $i$'s value $v_i \leq |X| (f+1) + 1$, the auctioneer reveals all its bids.
All feasible subsets contained in $X \cup Y \cup \{i\}$ that does not include $i$ has a cardinality at most $|X|$.
Thus, the virtual-value-optimal allocation allocates the set $X$ and the bidder $i$ as long as it has a positive virtual value.
Thus, bidder $i$'s critical bid equals $1$, the same as when the auctioneer is honest.

Now suppose that $v_i > |X| (f+1) + 1$.
Then, the auctioneer conceals all bidders $x \in X$, and pays a penalty $|X| \times f$.
However, $\{i\}$ is the virtual-value-optimal allocation with a virtual value at least $\vv(|X| (f+1) + 1) = |X| (f+1)$, against the total virtual value equal to $|Y| (f+1) = |X| (f+1)$ of the bidders in $Y$.
Further, bidder $i$'s critical bid equals $|X| (f+1) + 1$.
The auctioneer receives a net revenue equal to $|X| (f+1) + 1 - |X| \times f = |X| + 1 > 1$, the revenue from the honest strategy.
Thus, the strategic deviation generates a strictly larger revenue than the honest strategy, and thus, the DRA is not credible.

It only remains to show the existence of the bidder $i$ and the sets $X$ and $Y$ for an arbitrary non-matroid feasibility constraint. 

\begin{lemma} \label{lem:non-matroid}
    For any downward-closed feasibility constraint $\Feasibility$ that is not a matroid, there exists two feasible sets $\hat{X}, Y \in \Feasibility$ such that $\hat{X} \cap Y = \emptyset$, $|\hat{X}| = |Y| + 1$, and for any feasible subset $S$ of $\Feasibility$ contained in $\hat{X} \cup Y$, $S \neq \hat{X}$, then $|S| < |\hat{X}|$.
\end{lemma}
For any $i \in \hat{X}$, $X = \hat{X} \setminus \{i\}$ and $Y$ satisfy these properties of the sets $X$ and $Y$.
\end{proof}

\begin{proof}[Proof of \Cref{lem:non-matroid}]
    Let the constraint $\Feasibility$ be defined over the set of elements $E$.
    Suppose that $|E| \geq 3$ (all feasibility constraints over $2$ or less elements are matroids).
    It is not hard to verify that the claim holds when $E = 3$.
    We will prove the claim by induction on $e$, the number of items in $E$.

    Let the lemma be true for all non-matroid constraints over $e$ elements for all $e \leq n$.
    We will prove the claim when $|E| = n+1$.

    Since $\Feasibility$ is not a matroid, there exists some sets $A, B$ such that the augmentation property is not satisfied, i.e, $|A| > |B|$ and $B \cup \{i\} \not \in \Feasibility$ for any $i \in A$.
    Let 
    $$\mathcal V = \{(A,B) : A,B\in \Feasibility, |A|>|B|, B \cup \{e\} \notin \Feasibility~\forall e \in A \setminus B\}$$
    be the collection of all such sets $A$ and $B$.

    Choose $(X_0, Y_0) \in \mathcal{V}$ to maximize $|X_0 \cap Y_0|$ amongst all pairs in $\mathcal{V}$.
    We handle the following three cases separately --- either $X_0 \cap Y_0 \neq \emptyset$, or $X_0 \cap Y_0 = \emptyset$ and $|X_0| > |Y_0| + 1$, or $X_0 \cap Y_0 = \emptyset$ and $|X_0| = |Y_0| + 1$.
    \begin{enumerate}
        \item $X_0 \cap Y_0 \neq \emptyset$: Consider the collection $\hat{\Feasibility}$ of all feasible sets contained in $\Feasibility \setminus (X_0 \cap Y_0)$.
        Observe that $\hat{\Feasibility}$ is also not a matroid constraint since $X_0 \setminus (X_0 \cap Y_0)$ and $Y_0 \setminus (X_0 \cap Y_0)$ contained in $\hat{\Feasibility}$ also does not satisfy the augmentation property.
        By the induction hypothesis, there exists $\hat{X}, Y \in \hat{\Feasibility}$ which satisfies all the necessary properties.
        By construction, $\hat{X}$ and $Y$ are also feasible under $\Feasibility$ as required.
        \item $X_0 \cap Y_0 = \emptyset$ and $|X_0| > |Y_0| + 1$: For some element $i \in X_0$, consider the collection $\hat{Feasibility}$ of feasible sets contained in $E \setminus \{i\}$.
        $\hat{\Feasibility}$ is still not a matroid yet, since $X_0 \setminus \{i\}$, $Y$ continue to violate the augmentation property.
        By induction hypothesis, there exists $\hat{X}, Y \in \hat{\Feasibility}$ satisfying all the required properties.
        \item $X_0 \cap Y_0 = \emptyset$ and $|X_0| = |Y_0| + 1$: We will set $\hat{X} = X_0$ and $Y = Y_0$.
        We only have to verify that any feasible subset contained in $X_0 \cup Y_0$ that is not equal to $X_0$ has a cardinality strictly less than $X_0$.
        Assume otherwise and suppose that there exists $S$ such that $S \subseteq X_0 \cup Y_0$ and $|S| \geq X_0$.
        Note that $(S, Y_0) \in \mathcal{V}$.
        Further, $S \cap Y_0 > 0$.
        Otherwise, $S$ is contained in $X_0 \cup Y_0$, has cardinality at least $X_0$ and does not intersect with $Y_0$, and thus, $S$ must be equal to $X_0$.
        However, we choose $(X_0, Y_0) \in \mathcal{V}$ that maximizes the intersection between the two sets.
        Contradiction.
        Thus, any feasible subset contained in $X_0 \cup Y_0$ must either be $X_0$ or have a cardinality strictly smaller than $|X_0|$, as required. \qedhere
    \end{enumerate}
\end{proof}

\subsection{Omitted Proofs in \Cref{sec:LowerMHR}} \label{sec:proofofkkMHR}

\begin{proof}[Proof of \Cref{thm:nnMHR}]
As an immediate corollary to \Cref{thm:11MHR}, we can see that the largest monopoly reserve is also necessary in single-item multi-agent environments.
Consider a bidder with its value drawn from the exponential distribution (called the large bidder) and $(n-1)$ bidders with a deterministic value $\eta$ for a sufficiently small $\eta$ (called small bidders).
The revenue from being honest can be upper bounded by the revenue of an environment with two items --- the first item can solely be allocated to the large bidder and the second item is always allocated to one of the small bidders.
The auctioneer's expected revenue from allocating the item to the large bidder at a price equal to the monopoly reserve $1$ equals $\frac{1}{e}$.
Therefore, the expected honest revenue is at most $\frac{1}{e}+\eta$.
Choose $\eta$ to be sufficiently small such that the auctioneer's revenue while playing the strategy described in \Cref{thm:11MHR} allocating only the large bidder is greater than $\frac{1}{e} + \eta$.
Then, a collateral equal to $1 - \epsilon$ is not sufficient for single-item multi-agent environments.    
\end{proof}

For \Cref{thm:kkMHR} and \Cref{thm:1nMHR}, at a high level, the proofs are similar to \Cref{thm:11MHR} and consist of considering the difference in the auctioneer's revenue from being strategic and being honest when the bidders' values belong to different intervals (albeit requiring a more careful analysis for these environments).

\begin{proof}[Proof of \Cref{thm:kkMHR}]
Consider the auctioneer's strategy where it inserts a fake bid equal to $1+\delta$ and hides the bid when all bidders have values larger than $1$, but there exists a bidder $i$ with value $v_i \in [1, 1+\delta]$.
It reveals the fake bid in all other scenarios.

The auctioneer receives a revenue $k \, (1+\delta)$ when all bidders have a value larger than $1+\delta$, and $1$ per allocated bidder even if a single bidder has a value in the range $[0, 1+\delta]$.
The auctioneer pays a penalty $1-\epsilon$ when (a) there are a positive number of bidders with values between $1$ and $1+\delta$ and (b) all bidders have values that beat the monopoly reserve (this is exactly when the auctioneer hides its fake bid).

Once again, we will evaluate the difference in the auctioneer's revenue when it is strategic and when it is honest.
In  both the strategies, the auctioneer receives a revenue of $1$ per allocated bidder irrespective of the allocated bidders' values.
Further, the strategic auctioneer receives an extra $k\delta$ revenue when all bidders have a value at least $(1+\delta)$, which happens with a probability $e^{-k\, (1+\delta)}$.
The auctioneer pays a penalty precisely when all bidders have value greater than $1$ (probability of which equals $e^{-k}$), but not all of them have a value greater than $1+\delta$ (which happens with a probability $e^{-k(1+\delta)}$).
Thus, the auctioneer's expected marginal revenue from being strategic equals $k \delta e^{-k \, (1+\delta)} - (1-\epsilon) \, (e^{-k} - e^{-k \, (1+\delta)})$
which is positive if and only if
$$\frac{k\delta \, e^{-k \, (1+\delta)}}{(e^{-k} - e^{-k \, (1+\delta)})} \geq (1-\epsilon).$$
Once again, the limit of the left hand side as $\delta \xrightarrow{} 0$ is $1$, and thus, for a small enough $\delta$, the strategic auctioneer receives a larger revenue than being honest.    
\end{proof}

\begin{proof}[Proof of \Cref{thm:1nMHR}]    
    Similar to the single-item single agent environment, the auctioneer inserts a fake bid equal to $1+\delta$, and conceals the fake bid if the highest bid is in the range $[1, 1+\delta]$.

    We break down the proof into five cases depending on the values of the largest and the second largest bids:
    \begin{enumerate}
        \item Both the largest and the second largest bids are greater than $1+\delta$: In this case, the presence of the fake bid is irrelevant, and the auctioneer generates the same revenue irrespective of whether it is honest or strategic.
        \item The largest bid is greater than $1+\delta$, but the second largest bid is in the range $[1, 1+\delta]$:
        The auctioneer receives a payment $1+\delta$ when being strategic, but receives the second highest bid (smaller than $1+\delta$) when being truthful.
        However, as the number of bidders $n \xrightarrow{} \infty$, it is very likely that the second highest bid is quite close to $1+\delta$.
        Thus, we will pretend that the auctioneer receives the same revenue both, when being strategic and when being honest.
        \item The largest bid is greater than $1+\delta$ and the second largest bid is smaller than $1$: The auctioneer makes $1+\delta$ when being strategic, while it makes $1$ by being honest.
        Thus, the auctioneer makes $\delta$ more by being strategic in this case, which happens with a probability $n \, e^{-(1+\delta)}(1-e^{-1})^{(n-1)}$.
        \item The largest bid is in the range $[1, 1+\delta]$ (irrespective of the value of the second largest bid): The auctioneer receives the same payment when being strategic and when being honest.
        However, the auctioneer has to pay a penalty $(1-\epsilon)$ since it has to conceal its fake bid.
        This case happens with a probability $(1-e^{-(1+\delta)})^n - (1-e^{-1})^n$ (i.e, the difference between the probabilities of all bids being smaller than $1+\delta$ and smaller than $1$, respectively).
        \item Both, the largest and the second largest bids are smaller than $1$: No bids beat the reserve and thus, no bidder is allocated the item.
        The auctioneer does not have to conceal its fake bid, since no bidder will get allocated anyways.
    \end{enumerate}

    Thus, the expected difference in revenues equals $$\delta ne^{-(1+\delta)}(1-e^{-1})^{n-1} - (1-\epsilon) \, \Big[(1-e^{-(1 + \delta)})^n - (1-e^{-1})^n \Big]$$
    which is greater than zero if and only if
    $$\frac{\delta n e^{-(1+\delta)}(1-e^{-1})^{n-1}}{\Big[(1-e^{-(1 + \delta)})^n - (1-e^{-1})^n \Big]} > 1-\epsilon.$$
    As was the case in the single-item single-agent environment, the left hand side tends to $1$ as $\delta$ approaches zero.
    Thus, the auctioneer can inject a bid sufficiently close to $1$ to extract larger revenue by being strategic than while playing the honest strategy.
\end{proof}

\subsection{Proof of \Cref{thm:PrivatePublicSep}}
\label{appendix:proof-of-PrivatePublicSep}

Consider the following strategy by the auctioneer.
The auctioneer does not communicate any of the bids placed by the bidders to each other.
The auctioneer sends $(k-1)$ large bids, all greater than $1$ to each of the $k$ bidders and reveals all of them irrespective of the bids submitted by the agents.
Finally, the auctioneer communicates a bid equal to $1+\delta$ to all users, and conceals the bid from any bidder who realizes a value between $1$ and $1+\delta$.

From the perspective of each bidder, $k-1$ items are sold to the $k-1$ large fake bids inserted by the auctioneer and the bidder is competing for the final available item.
Thus, the auctioneer has reduced the auction into $k$ parallel single item auctions, with two bidders in each of them (one real bidder and the other, the fake bid $1+\delta$). However, the total collateral put up by the fake bid is $1$, and not $1$ each for the $k$ parallel auctions.

The auctioneer's revenue from being strategic can be lower bounded as follows.
The auctioneer always pays a penalty $1$, to get a ``license'' to cheat with its fake bid.
With the penalty paid upfront, the auctioneer makes a revenue pointwise larger than the optimal auction --- $\delta$ more whenever a bidder has a value larger than $1+\delta$ (which happens with a probability $e^{-(1+\delta)}$), and exactly the same as the optimal auction at other values.
Thus, the expected additional revenue from inserting a fake bid with value $1+\delta$ equals $k \, \delta e^{-(1+\delta)}$.

For a sufficiently large $k$, the additional revenue is larger than the ``license fee'' $1$ paid by the auctioneer, and thus, it is profitable to deviate from the honest strategy.

\subsection{Proof of \Cref{thm:ADRA}}
\label{appendix:proof-of-ADRA}


To analyze credibility of the ADRA, we will make the following simplifying modification.
When the set of $\mathsf{active}$ bidders have been trimmed down sufficiently so that all bidders in $\activeBid_{\ell}$ can be simultaneously allocated by the ADRA, we break the loop in \Cref{bul:ADRALoop} in \Cref{def:ADRA} (i.e, the original version) to stop the auction (let this happen during level $\ell$).
We then allocate all $\mathsf{active}$ bidders at the final price, and ask them to reveal their bids in level $\ell + 1$.
An alternative could be to continue with the loop in \Cref{bul:ADRALoop}, increasing the virtual price until it exceeds all virtual bids and allocate all $\mathsf{promised}$ bidders in the final simulation of $\hat{\mechanism}$ before all bidders $\mathsf{quit}$.
The alternate design ensures that all bidders reveal their bids only after the virtual price has exceeded the virtual value of their bids.
Note that this is not the case in the original version, since bidders allocated during the final stage (\Cref{bul:ADRAFinal}) might be asked to reveal their bids even though the virtual price is smaller than their virtual bids.

The two variants of the ADRA are strategically equivalent to the auctioneer.
Let $\activeBid_{\ell}$ be the set of active bidders at the start of level $\ell + 1$.
For any set of bids $\mathsf{aborted}$ by the auctioneer during level $\ell + 1$ (or after, in case of the modified version), the payment charged to any bidder is independent of the bids of the bidders in $\activeBid_{\ell}$, as long as they are large enough to remain $\mathsf{active}$ during round $\ell + 1$.
All virtual bids not belonging to $\activeBid_{\ell}$ are strictly smaller than the virtual bids belonging to $\activeBid_{\ell}$.
Since $\hat{\mechanism}$ optimizes virtual-surplus, and the feasibility constraint is a matroid, all bidders in $\activeBid_{\ell}$ are allocated irrespective of the bids $\mathsf{aborted}$ by the auctioneer.
Thus, the critical bids of these bidders are independent of the other bids in $\activeBid_{\ell}$, and is purely a function of the bids $\mathsf{aborted}$ by the auctioneer in levels $\ell + 1$ (or after, in the modified version) and the bids already revealed by the bidders that have $\mathsf{quit}$ in previous levels.
Similarly, conditioned on allocating all bidders in $\activeBid_{\ell}$, the critical bids of all $\mathsf{promised}$ bidders are also purely a function of the bids $\mathsf{aborted}$ by the auctioneer and previously revealed bids.
Therefore, learning the bids in $\activeBid_{\ell}$ during level $\ell + 1$ in the original version does not allow the auctioneer to increase its revenue.
On the other hand, $\mathsf{aborting}$ bids in levels after $\ell + 1$ in the modified game also does not help the auctioneer, since the auctioneer could have extracted a strictly larger revenue by $\mathsf{aborting}$ those bid during level $\ell + 1$.

While the original version achieves a better communication complexity, since it terminates in a lesser number of levels, it is convenient to show that the modified version is credible.
In the following, we show credibility of ADRA for matroids by proving any deviation $G'$ for the auctioneer cannot improve the revenue.



But first, let us define the suggested strategy $S_i^\ADRA(b_i)$ for bidder $i$ as follows:
\begin{itemize}
    \item Commit to bidding $b_i$ (which may differ from her true value $v_i$).
    \item $\mathsf{Quit}$ at the first level $\ell$ where $\overline{\vv_i}(b_i) < p_\ell^\vv$ and reveal the commitment.
\end{itemize}

Observe that bidders should always use some suggested strategy, as other strategies always end up with a negative utility, and it suffices for us to consider suggested strategies only.

\begin{lemma}\label{lemma:use-suggested-strategy}
    It is always a best response for a real bidder $i \in [n]$ to implement the suggested strategy $S_i^\ADRA(b_i)$ for some $b_i$.
\end{lemma}

\begin{proof}
    Whenever bidder $i$ deviates from suggested strategies, they will not be allocated and lose their deposit according to the perfect binding property of the commitment scheme, resulting a negative utility. Therefore, it is dominated by the truthful strategy $S_i^\ADRA(v_i)$ which always yields a non-negative utility.
\end{proof}

Also note that a bidder will always be allocated once $\mathsf{promised}$ with a fixed payment.

\begin{lemma}\label{lem:clinched-implies-allocation}
    A real bidder $i \in [n]$ using a suggested strategy $S_i^\ADRA(b_i)$ for some $b_i$ will always be allocated once $\mathsf{promised}$ with a payment that only depends on auctioneer and bidders' actions before she becomes $\mathsf{promised}$.
\end{lemma}

\begin{proof}
    To see that bidder $i$ will always be allocated once $\mathsf{promised}$ by the end of some level $\ell$ with a payment $p$, note that $\overline{\vv_i}(p)\le p_\ell^\vv$ and every bidder $j$ that does not $\mathsf{quit}$ or $\mathsf{abort}$ by the end of the level must have committed to a bid $\overline{\vv_j}(b_j) > p_\ell^\vv$ or eventually $\mathsf{abort}$, and thus cannot affect bidder $i$'s threshold price.
\end{proof}

Given any auctioneer's deviation $G'$, a strategy profile $\vb{S}$, and bidders' value $\vb{v}$,  denote the final allocation by $A(G',\vb{S}(\vb{v})) \in \hat{\Feasibility}$.  
When it is clear from context, we may abbreviate this as $A(G')$.
In the following, we assume every real bidder $i \in [n]$ implements a suggested strategy $S_i^\ADRA(b_i)$ for some $b_i$, i.e., $\vb{S}=\vb{S}^\ADRA$.

We first show the auctioneer never $\mathsf{abort}$ fabricated bids that is $\mathsf{promised}$ using any undominated deviation $G'$.

\begin{lemma}\label{lem:undominated-deviation}
    When the auctioneer employs any undominated deviation $G'$, a fabricated bid $b_i$ is never $\mathsf{aborted}$ at some level if it is $\mathsf{promised}$ or it would be $\mathsf{promised}$ by the end of the level if not $\mathsf{aborted}$.\footnote{Here we implicitly assume the auctioneer always act last at each level and decide whether to $\mathsf{quit}$ or $\mathsf{abort}$ fabricated bids after receiving all real bidders' actions. \label{footnote:auctioneer-action-last}}
\end{lemma}

\begin{proof}
Consider a deviation $G'$ where the auctioneer $\mathsf{abort}$ a fabricated bid $b_i$ at some information set $I$ of level $\ell$ where $b_i$ is $\mathsf{promised}$ or going to be $\mathsf{promised}$.
Let $\ell'\le\ell$ be the last level where $b_i$ is not $\mathsf{promised}$ at the beginning of the level.
Define $G''$ be another deviation that is identical to $G'$ except that $G''$ does not $\mathsf{abort}$ $b_i$ at information set $I$, and after this point, $G''$ proceeds exactly as if $b_i$ had been $\mathsf{aborted}$ in $G'$.\footnote{This is feasible because the auctioneer's subsequent decisions remain unaffected by whether a fabricated bid is $\mathsf{aborted}$.}

Let $B',B''$ denote sets of bidders who were not $\mathsf{aborted}$ under $G'$ and $G''$, respectively.
By definition, for any strategy profile $\vb{S}$ and bidders' value $\vb{v}$, $G''$ reaches the information set $I$ if and only if $G'$ reaches $I$, and the final set of bidders who did not $\mathsf{abort}$ will be identical expect for $i$, i.e., $B'' = B' \cup \{i\}$.
By the optimality of ADRA, the allocations $A(G'')$ and $A(G'')$ correspond to the virtual-surplus-maximizing feasible subsets of $B'$ and $B''$, respectively, and $i \in A(G'')$ by \Cref{lem:clinched-implies-allocation}.
Since $\hat{\Feasibility}$ is a matroid, the size of $G''$ and $G'$ differ by at most $1$, and we consider these two cases separately:
\begin{enumerate}
    \item When $|A(G'')| = |A(G')|+1$, we have $A(G'') = A(G') \cup \{i\}$.
    Then, compared to $G''$, $G'$ loses the collateral $d_{\ell'}$ for aborting $b_i$ without getting anything in return.
    \item When $|A(G'')| = |A(G')|$, there exists some $j \in [\hat{n}]$ such that $\overline{\vv_j}(b_j)\le\overline{\vv_i}(b_i)$ and $(A(G') \setminus \{j\}) \cup \{i\} \in \hat{\Feasibility}$ by \Cref{thm:MatroidBasisExchange} and optimality of $A(G'')$. 
    Then, the difference of revenue between $G'$ and $G''$ will be at most $b_j-d_{\ell'}$ (when $j$ is a real bidder), which must be a negative number since otherwise $b_i$ would not be $\mathsf{promised}$ at the beginning of level $\ell'+1$ (or before aborting if $\ell'=\ell$).
\end{enumerate}
In conclusion, $G''$ dominates $G'$. 
\end{proof}

We further show allocation monotonicity between two suggested strategies $S^\ADRA_i(b_i), S^\ADRA_i(b_i')$ with $b_i < b_i'$ via \Cref{lem:clinched-before-quit} and \Cref{lem:allocation-monotonicity}.

\begin{lemma}\label{lem:clinched-before-quit}
   When the auctioneer employs any undominated deviation $G'$, consider two suggested strategies $S_i^\ADRA(b_i), S_i^\ADRA(b_i')$ for a real bidder $i \in [n]$ with $b_i < b_i'$.
    Whenever bidder $i$ $\mathsf{quit}$ at some level $\ell$ and is eventually allocated using $S_i^\ADRA(b_i)$, she must be $\mathsf{promised}$ by the end of level $\ell$ using $S_i^\ADRA(b_i')$.
\end{lemma}

\begin{proof}
    Assume for some strategy profile $\vb{S}$ and bidders' value $\vb{v}$,  bidder $i$ $\mathsf{quit}$ at level $\ell$ but ends up being allocated (i.e., $i \in A(G')$) using $S_i^\ADRA(b_i)$, while she would not be $\mathsf{promised}$ by the end of level $\ell$ when using $S_i^\ADRA(b_i')$ with $b_i'>b_i$.
    Let $i_1,i_2,...,i_k \in [\hat{n}]$ be the bidders who $\mathsf{abort}$ after bidder $i$ $\mathsf{quit}$, listed in the order in which they abort.
    For every $0 \le j \le k$, further define $A^{(j)} \in \hat{\Feasibility}$ to be the final allocation assuming only a prefix of these bidders $\{i_1,i_2,\ldots,i_j\}$ $\mathsf{abort}$, and we have:
    \begin{itemize}
        \item $i \notin A^{(0)}$ since bidder $i$ would not be $\mathsf{promised}$ by the end of level $\ell$ when using $S_i^\ADRA(b_i')$ with $b_i'>b_i$ (which is  indistinguishable from $S_i^\ADRA(b_i)$ before she $\mathsf{quit}$ at level $\ell_i$ due to the perfectly hiding property of the commitment scheme).
        \item $i \in A^{(k)}$ since $A(G') = A^{(k)}$.
    \end{itemize}
    
    By a hybrid argument, there exists some $j \in [k]$ such that $i \notin A^{(j-1)}$ while $i \in A^{(j)}$.
    By the optimality of ADRA and the matroid augmentation property, it must be the case that $(A^{(j-1)} \setminus \{i_j\}) \cup \{i\} = A^{(j)}$. However, that further implies the fabricated bidder $i_j$ has already become $\mathsf{promised}$ before being $\mathsf{aborted}$, due to the optimality of $A^{(j)}$ and the fact that $i$ is already $\mathsf{quit}$. By \Cref{lem:undominated-deviation}, this gives us a contradiction to $G'$ being undominated.
\end{proof}

\begin{lemma}\label{lem:allocation-monotonicity}
    When the auctioneer employs any undominated deviation $G'$, consider two suggested strategies $S_i^\ADRA(b_i), S_i^\ADRA(b_i')$ for a real bidder $i \in [n]$ with $b_i < b_i'$.
    Whenever bidder $i$ is allocated using $S_i^\ADRA(b_i)$, she will also be allocated using $S_i^\ADRA(b_i')$ with an identical payment.
\end{lemma}

\begin{proof}
    By \Cref{lem:clinched-before-quit}, bidder $i$ using $S_i^\ADRA(b_i)$ can be allocated only if she is $\mathsf{promised}$ by the end of level $\ell$ (and eventually allocated) when using $S_i^\ADRA(b_i')$. 
    Denote bidder $i$'s payment when using $S_i^\ADRA(b_i)$, $S_i^\ADRA(b_i')$ by $p_i, p_i'$, correspondingly.
    Note that $p_i \le b_i$ and $p_i' \le b_i'$.
    We want to show $p_i=p_i'$.


    Assume $p_i<p_i'$.
    It must be the case that the auctioneer $\mathsf{abort}$ some fabricated bid $b_j$ such that $b_j = p_i'$ after bidder $i$ $\mathsf{quit}$ at level $\ell$ when using $S_i^\ADRA(b_i)$.
    When $b_j < b_i$, not aborting it will save its collateral $d_\ell$ while increasing the payment from bidder $i$;
    when $b_j \ge b_i$, the collateral $d_\ell$ is saved while losing the payment of from bidder $i$, still resulting a positive increment to the total revenue as $p_i \le b_i < d_\ell$.
    In conclusion, aborting $b_j$ is always dominated by not aborting $b_j$, which contradicts to $G'$ being undominated. 

    Now assume $p_i>p_i'$.
    It must be the case that the auctioneer $\mathsf{abort}$ some fabricated bid $b_j$ such that $b_j = p_i$ at level $\ell$ if bidder $i$ does not $\mathsf{quit}$ using $S_i^\ADRA(b_i')$.
    But that implies $b_j\le b_i <b_i'$, and hence not aborting it will save its collateral $d_\ell$ while increase the payment from bidder $i$.
    Therefore, aborting $b_j$ is always dominated by not aborting $b_j$, which contradicts $G'$ being undominated. 
\end{proof}

Now we are ready to prove that truth-telling forms an ex-post Nash no matter which undominated deviation $G'$ the auctioneer employs.

\begin{lemma}\label{lem:ADRA-equilbrium}
    When the auctioneer employs any undominated deviation $G'$, the truthful strategy $\vb{S}^\ADRA(\vb{v})=(S_1^\ADRA(v_1),S_2^\ADRA(v_2),\ldots,S_n^\ADRA(v_n))$ forms an ex-post Nash equilibrium.
\end{lemma}

\begin{proof}
    Consider the best response for a real bidder $i \in [n]$ against auctioneer using $G'$ and other bidders using $\vb{S}_{-i}^\ADRA(\vb{v}_{-i})$.
    By \Cref{lemma:use-suggested-strategy}, it suffices for us to consider $S_i^\ADRA(b_i)$ for $b_i \ne v_i$ and show $S_i^\ADRA(v_i)$ weakly dominates all of them.

    When $b_i<v_i$, we immediately have $S_i^\ADRA(v_i)$ weakly dominates $S_i^\ADRA(b_i)$ by \Cref{lem:allocation-monotonicity}.
    On the other hand, when $b_i>v_i$, if bidder $i$ gets allocated using $S_i^\ADRA(v_i)$, then $S_i^\ADRA(b_i)$ will also end up with a same payment by \Cref{lem:allocation-monotonicity}.
    The only case left to analyze is where bidder $i$ only gets allocated by using $S_i^\ADRA(b_i)$ but not when using $S_i^\ADRA(v_i)$.

    Assume the payment $p$ for $S_i^\ADRA(b_i)$ satisfies $p \le v_i$.
    Let $\ell$ be the level that bidder $i$ $\mathsf{quit}$ when using $S_i^\ADRA(v_i)$, and we know bidder $i$ will never be $\mathsf{promised}$ in that case by \Cref{lem:clinched-implies-allocation}.
    Therefore, the auctioneer have to $\mathsf{abort}$ some fabricated bid $b_j$ such that $b_j > v_i$ after observing that bidder $i$ does not $\mathsf{quit}$ at level $\ell$. However, not aborting $b_j$ can save a collateral of at least $b_j$ while losing a revenue of $p \le v_i < b_j$, which contradicts $G'$ begin undominated. Therefore, the payment $p$ for $S_i^\ADRA(b_i)$ satisfies $p > v_i$, which yields a negative utility for bidder $i$.
\end{proof}

\begin{proof}[Proof of \Cref{thm:ADRA}]
    Recall that \Cref{lem:ADRA-equilbrium} states that any undominated deviation $G'$ is also an BIC mechanism. Since it is already shown that ADRA for matroid is revenue-optimal among BIC mechanisms, we conclude that no deviation can provide more revenue than following the protocol exactly, and thus ADRA for matroids is credible.
\end{proof}

%
%
%
%
%

\end{document}